\documentclass{article}
\usepackage{graphicx}
\usepackage{amsmath}
\usepackage{amssymb}
\usepackage{amsthm}
\usepackage{multicol}
\usepackage{multirow}
\usepackage{subcaption}
\usepackage{url}
\usepackage[mathscr]{euscript}
\usepackage{array}
\usepackage{float}
\usepackage{xspace}
\usepackage{comment}
\usepackage{tabularx}
\usepackage{todonotes}
\usepackage{bm}
\usepackage{enumitem}
\usepackage{hyperref}
\usepackage{tikz}
\usepackage{thmtools}
\usepackage{mathrsfs}
\usepackage{thm-restate}
\usepackage{stmaryrd,nicefrac,mathtools}

\usepackage{amsthm,amsmath,amssymb,amsfonts,multicol}
\usepackage{enumitem,caption}

\newcommand{\tuple}[1]{{\langle #1 \rangle}}

\newcommand{\domai}[1]{W^{#1}}

\newcommand{\Model}{{\mathcal M}}
\newcommand{\Frame}{{\mathcal F}}

\newcommand{\ax}[1]{{\rm #1}}
\newcommand{\val}[1]{\|#1\|}
\newcommand{\peq}{\preccurlyeq}
\newcommand{\seq}{\succcurlyeq}

\newcommand{\R}{\mathrel R}

\newenvironment{proofidea}{%
    \par\noindent{\it Proof idea.}\hspace{0.1em}}%
    {\hfill$\square$\par\vspace{1ex}}

\usepackage{tikz}
\usetikzlibrary{backgrounds}
\usetikzlibrary{arrows}
\usetikzlibrary{shapes,shapes.geometric,shapes.misc}
\usetikzlibrary{patterns}

\tikzstyle{tikzfig}=[baseline=-0.25em,scale=0.5]

\pgfkeys{/tikz/tikzit fill/.initial=0}
\pgfkeys{/tikz/tikzit draw/.initial=0}
\pgfkeys{/tikz/tikzit shape/.initial=0}
\pgfkeys{/tikz/tikzit category/.initial=0}

\pgfdeclarelayer{edgelayer}
\pgfdeclarelayer{nodelayer}
\pgfdeclarelayer{fillbacklayer}
\pgfdeclarelayer{filltoplayer}
\pgfsetlayers{background,edgelayer,fillbacklayer,filltoplayer,nodelayer,main}

\tikzstyle{none}=[inner sep=0mm]

\tikzstyle{every loop}=[]

\tikzset{
    root/.style={
        circle,                
        draw=black,            
        fill=white,            
        minimum size=6mm,      
        inner sep=1pt,         
        font=\small            
    }
}

\tikzset{
    sroot/.style={
        circle,                
        draw=black,            
        fill=white,            
        minimum size=2mm,      
        inner sep=1pt,         
        font=\small            
    }
}

 \tikzstyle{edgeto}=[->]

\newtheorem{theorem}{Theorem}[section]

\newtheorem{definition}[theorem]{Definition}
\newtheorem{lemma}[theorem]{Lemma}

\newtheorem{corollary}[theorem]{Corollary}
\newtheorem{proposition}[theorem]{Proposition}

\newtheorem{remark}[theorem]{Remark}
\newtheorem{example}[theorem]{Example}

\begin{document}

\title{Gödel-Tarski Translation for Constructive Logics with Transitive Closure}

\author{Sofía Santiago-Fernández\thanks{\url{sofia.santiago@ub.edu}} \and David Fernández-Duque\thanks{\url{fernandez-duque@ub.edu}} \and Joost J. Joosten\thanks{\url{jjoosten@ub.edu}}}

\date{}

\maketitle

\begin{abstract}
We introduce the semantically-defined constructive propositional dynamic logics of regular programs $\sf CPDL$ and ${\sf WPDL}$, respectively extending the basic constructive modal logic $\sf CK$ \cite{Bellin2001ExtendedCC} and its Wijesekera-style variant $\sf WK$~\cite{wijesekera1990constructive}.
Using translations between our logics and $\sf PDL$ \cite{FISCHER1979194}, we show that both $\sf CPDL$ and $\sf WPDL$ are {\sc ExpTime}-complete and enjoy an exponential-size finite model property. 
In particular, their diamond-free fragment, restricted to a single atomic program and its reflexive transitive closure is already {\sc ExpTime}-complete, thereby settling the conjecture of Afshari et al~\cite{AfshariGLZ24}.
By embedding $\sf CS4$ and $\sf WS4$ into our logics, we show their validity problems are also in {\sc ExpTime}, improving the previously known {\sc NExpTime} upper bound obtained via finite model arguments~\cite{balbiani2024variants}.
\end{abstract}

\section{Introduction}

Propositional dynamic logic $\sf PDL$ \cite{FISCHER1979194}, grounded on modal logic and verification logic, is a well-es\-ta\-blished logic for reasoning about correctness, termination and equivalence of programs. Its central idea is to interpret programs $a$ as modal operators, thereby allowing one to reason about assertions $p$ after program execution. In particular, $\langle a \rangle p$ informally expresses that there exists an execution of $a$ that terminates with $p$ holding at termination, whereas $[a]p$ expresses that every terminating execution of $a$ ends in a state where $p$ holds. These operators allow us, for instance, to express termination by $\langle a \rangle \top$ and determinacy with respect to an assertion $p$ by $\langle a \rangle p \to [a]p$. The language of $\sf PDL$ provides several operations for constructing complex programs: non-deterministic choice $a \ \cup \ b$, sequential composition $a \ ; \ b$, iteration $a^*$, and tests of assertions $p?$ which represent actions that can be executed only when $p$ holds. These operations make it possible to encode familiar programming constructs within the logic e.g. `if $p$ then $a$ else $b$' can be represented by $(p?\ ;\ a) \cup (\neg p? \ ; \ b)$, and `while $p$ do $a$' as $(p? \ ; \ a)^* \ ; \ \neg p?$. Beyond its expressive power, $\sf PDL$ enjoys the finite model property and has an {\sc ExpTime}-complete validity problem \cite{FISCHER1979194,PRATT1980231}.

With this computational motivation in mind, it is only natural to investigate intuitionistic counterparts of $\sf PDL$. This has been pursued in the intuitionistic dynamic logic $\sf iPDL$~\cite{degen}, which may best be described as only `semi-constructive' since complex formulas do not necessarily satisfy the upwards-persistence condition of intuitionistic truth. Nishimura~\cite{NishimuraConstructivePDL} has also studied the diamond-free fragment of constructive $\sf PDL$, establishing a sequent calculus and the finite model property. This fragment has further found applications in proof theory, where it plays an auxiliary technical role~\cite{leivant}.

There are, however, several ways of defining modal logics over an intuitionistic base (see e.g.~\cite{BalbianiGGO24,Grefe96,Simpson94,fitch1948intuitionistic,Wolter1997}), with the precise axioms and frame conditions varying wildly depending on the intended applications.
In particular, the `constructive' variants~\cite{Prawitz1965,Bellin2001ExtendedCC} are motivated by computational considerations and provide extensions of the Curry-Howard correspondence~\cite{DaviesP96}. 
In this setting, worlds in Kripke semantics may be viewed as stages of computation, while modal propositions correspond to types describing code available or executable at those stages. 
From this perspective, the intuitionistic accessibility relation $\peq$ represents context extension, whereas the modal relation represents transitions between computational stages. 
The basic constructive modal logic is known as $\sf CK$ \cite{Bellin2001ExtendedCC} (`constructive $\sf K$'), and its only modal axioms are two variants of the $\sf K$ axiom, namely, $\Box(p\to q)\to(\Box p\to \Box q)$ and $\Box(p\to q)\to(\Diamond p\to \Diamond q)$.
By adding $\neg\Diamond\bot$, we obtain the Wijesekera logic $\sf WK$~\cite{wijesekera1990constructive}, and if we instead add $\Box p\to\Box\Box p$ and $\Diamond\Diamond p\to \Diamond p$, we obtain the well-studied logic $\sf CS4$ ~\cite{AlechinaMPR01}.
A distinctive feature of constructive modal logics is that $\Diamond$ and $\Box$ are not inter-definable: in particular, $\neg\Diamond\neg\varphi$ is not equivalent to $\Box\varphi$. Consequently, both modalities must be treated as primitive operators.

In this work, we introduce the constructive counterpart of $\sf PDL$, denoted by $\sf CPDL$, and its Wijesekera variant $\sf WPDL$, extending $\sf CK$ and $\sf WK$, respectively. Based on the idea that truth should persist under context extension during program execution, this perspective is particularly natural for programs that interact with their environment (e.g. external sensors, user input, communication with other systems), where the information available during execution may grow. Accordingly, $\langle a\rangle p$ is read as saying that, under any context extension, $a$ can terminate with $p$ holding at termination; whereas $[a]p$ states that, under any context extension, every terminating execution of $a$ ends in a state where $p$ holds. 
In the diamond-free fragment, we show that our logic $\sf CPDL_\Box$ coincides with Nishimura's constructive $\sf PDL$ (Theorem~\ref{EqNishimura}). 
Moreover, restricting the language to a single atomic program $a$ and its iteration $a^*$ yields the modal logic ${\sf CK}^*_\Box$, for which sound and complete Hilbert~\cite{celani} (there denoted $\sf IK^*_\Box$) and cyclic~\cite{AfshariGLZ24} (there denoted $\sf IM_K$) calculi have been developed.
Similar logics, but with diamonds, have   been studied in a `temporal' setting, where the atomic action is deterministic.
In this context, the modality $[a^*]$ is read as `henceforth' and $\langle a^* \rangle$ as `eventually'.
However, these intuitionistic temporal logics involve non-constructive axioms and the best upper bounds on their complexity are non-elementary~\cite{Boudou2017}.

The semantics for constructive modal logics are defined over bi-relational models $\tuple{W,\peq ,\R, \|\cdot\|}$, where $\peq$ is used to model the intuitionistic implication, $\R$ is in principle an arbitrary relation on $W$ used to interpret the modalities, and $\|\cdot\|$ is a valuation satisfying standard conditions of atomic {\em ex falso quodlibet}, atomic and falsum persistence, and falsum seriality.
In order to maintain the `upwards persistence' of intuitionistic truth, modalities are interpreted via a combination of $\peq$ and $R$.
This works without a hitch for defining semantics for $\sf CK$, which enjoys a Gödel-Tarski translation into the fusion ${\sf S4}\oplus {\sf K} $~\cite{Dalmonte25} and hence inherits the finite model property and {\sc PSpace} decidability~\cite{Wolter1996Fusions}.

In order to adapt this semantic framework to $\sf PDL$, we need to models of the form $\langle W, \peq, \rho, \|\cdot\| \rangle$, where $\rho$ is a function that interprets atomic programs.
When $a$ is atomic, $[a]$ and $\langle a\rangle$ can be interpreted by combining $\peq$ and $\rho(a)$ as in the modal case without trouble.
The issue is that complex programs no longer behave as intended if we let $\rho(a)$ be arbitrary; for example, $[ a ; b ] \varphi \equiv [ a ] [ b ] \varphi$ does not hold in general, where `$;$' denotes composition (see Example~\ref{ExFormulaEqDiamondCPDL} and Corollary~\ref{FormulaEqDiamondCPDL}).
Intuitively, the reason for this is the combination of $\peq$ and $\rho (b)$, which comes into play when evaluating $[ a ] [ b ] \varphi$, but is lost in the composition $\rho(a;b)$.
Nishimura~\cite{NishimuraConstructivePDL} remedies this by demanding that ${\peq; \rho(a)} =
\rho(a)$ for all atomic programs $a$, which allows us to interpret $[a]$ via $\rho(a)$ alone (see Lemma~\ref{AlternativeBox}); we will refer to this as the {\em Nishimura condition}.
However, in a language with diamonds, this condition may lead to unwanted new validities, such as $\neg\neg \langle a\rangle p\to \langle a\rangle p$.
In order to remedy this tension between the desirable behaviour of boxes and diamonds, we will split the interpretation of programs into a `box version' $R_\Box(a)$ and a `diamond version' $R_\Diamond(a)$.
Our priority is for $R_\Diamond (a^*) $ to correctly interpret the transitive diamond of $\sf CS4$; fortunately, we do not require any interaction with $\peq$ to achieve this.

Based on this semantic construction, we establish Gödel--Tarski translations from constructive $\sf PDL$ into classical $\sf PDL$, therefore inheriting the finite model property and {\sc ExpTime}-completeness on their validity problems. Building on this idea, we present a translation methodology for our constructive $\sf PDL$, summarized in Figure~\ref{fig:DiagramLogics}.

\begin{figure}[h!]
    \centering
\begin{tikzpicture}
	\begin{pgfonlayer}{nodelayer}
		\node [style=none] (0) at (-0.25, 0) {$\sf CS4$};
		\node [style=none] (1) at (-0.25, 2) {$\sf CPDL$};
		\node [style=none] (2) at (2, 0) {$\sf WS4$};
		\node [style=none] (3) at (2, 2) {$\sf WPDL$};
		\node [style=none] (4) at (4, 2) {$\sf PDL$};
        \node [style=none] (5) at (-2, 2) {${\sf CPDL}_\Box$};
        \node [style=none] (6) at (-3.75, 2) {${\sf CK}_{\Box}^*$};
        \node [style=none] (7) at (-5.75, 2) {$\sf K^*$};
	\end{pgfonlayer}
	\begin{pgfonlayer}{edgelayer}
        \draw[draw=none] (0) -- node[midway, sloped] {$\xhookrightarrow[\text{Sec. \ref{CWS4}}]{\kappa}$} (1);
        \draw[draw=none] (2) -- node[midway, sloped] {$\xhookrightarrow[\text{Sec. \ref{CWS4}}]{\kappa}$} (3);
        \draw[draw=none] (1) -- node[midway, sloped] {$\xhookrightarrow[\text{Sec. \ref{secCW}}]{\omega}$} (3);
        \draw[draw=none] (3) -- node[midway, sloped] {$\xhookrightarrow[\text{Sec. \ref{sec:WKtoPDL}}]{\tau}$} (4);
        \draw[draw=none] (6) -- node[midway, sloped] {$\subseteq$} (5);
        \draw[draw=none] (7) -- node[midway, sloped] {$\xhookrightarrow[\text{Sec. \ref{sec:KtoCKbox}}]{\iota}$} (6);

        \draw[draw=none] (5) -- node[midway, sloped] {$\subseteq$} (1);
	\end{pgfonlayer}
\end{tikzpicture}

    \caption{Overview of the constructive $\sf PDL$ and the embeddings.}
    \label{fig:DiagramLogics}
\end{figure}

As we mentioned, the main difference between constructive and Wijesekera logics lies in
the acceptance of $\neg \Diamond \bot$ in the latter, which semantically amounts to demanding
that $\bot$ be false on all possible worlds (i.e., that models be infallible). We therefore first show that ${\sf CPDL}$ can be faithfully reduced to ${\sf WPDL}$ (Theorem \ref{CorrespondenceConstructiveAndWijesekera}), which simplifies translations into classical logic and allows us to work under infallibility without loss of generality. This reduction is not needed over the diamond-free fragment, as fallible and infallible models validate exactly the same formulas (Proposition \ref{EqCKWKBox}).

The G\"odel-Tarski translation for ${\sf WPDL}$ is partly inspired by the translation of $\sf iPDL$ into classical $\sf PDL$, from which decidability and the finite model property are obtained \cite{degen}. However, adapting this translation to our setting requires additional modifications to ensure the upward persistence of truth. This allows for ${\sf WPDL}$ to immediately inherit upper bounds from ${\sf PDL}$ (Section \ref{sec:UpperBoundsMML}), namely, the exponential model property and {\sc ExpTime} upper bound on the validity problem. Composing this translation with the embedding of $\sf CPDL$ into $\sf WPDL$, we obtain the same bounds for ${\sf CK}^*_\Box$, $\sf CPDL_\Box$ and $\sf CPDL$.

We also obtain matching lower bounds by embedding the classical fragment $\sf K^*$ (the $a, a^*$-fragment of classical $\sf PDL$ whose only programs are $a$ and its iteration $a^*$) into $\sf CK^*_\Box$ (Theorem~\ref{KiffCK}). This is achieved through a `converse G\"odel-Tarski'-style translation that internalizes classical reasoning: the translation of a $\sf K^*$-formula $\varphi$ holds in a constructive model precisely when, assuming the excluded middle for all subformulas of $\varphi$ at every accessible world, then $\varphi$ itself holds. Therefore, we conclude that $\sf CK^*_\Box$, ${\sf CPDL}_\Box$, ${\sf CPDL}$, and ${\sf WPDL}$ are all {\sc ExpTime}-complete. This in particular proves the conjecture of~\cite{AfshariGLZ24} for ${\sf CK}^*_\Box$.

As a final result we apply our framework to constructive modal logics by showing that  $\sf CS4$ can be translated into the $a^*$-fragment of ${\sf CPDL}$. Semantically, the $\sf CS4$ modalities can be viewed as standard modal operators interpreted over the composition $(\peq;R)$, which naturally suggests a translation interpreting $\Box$ and $\Diamond$ by iterative modalities $[a^*]$ and $\langle a^*\rangle$, respectively. Hence, $\sf CS4$ is also decidable in {\sc ExpTime}, improving the previously known {\sc NExpTime} upper bound obtained via finite model arguments~\cite{balbiani2024variants}. The same method yields an {\sc ExpTime} upper bound for the Wijesekera variant $\sf WS4$.

We next review classical $\sf PDL$, introducing its syntax, semantics, and the main results needed for our subsequent developments.

\subsection{Propositional dynamic logic of regular programs}\label{PDLdefs}

$\sf PDL$ is defined over the language $\mathcal L$ given by the following grammar, where formulas $\varphi$ and programs $\alpha$ are defined simultaneously:
\begin{align*}
    \varphi &:=  \bot \mid p \mid \varphi\wedge\varphi  \mid \varphi\vee\varphi \mid \varphi \to \varphi \mid [\alpha]\varphi \mid \langle \alpha \rangle \varphi \\
    \alpha & :=  \ a \mid  (\alpha;\alpha) \mid  (\alpha \cup \alpha)  \mid \alpha^* 
    \mid \varphi ?
\end{align*}

Here, $p$ ranges over a set of propositional variables $\rm Prop$ and $a$ over a set of atomic programs $\Pi_0$.
For the language $\mathcal{L}$ we write $\Pi_{\mathcal{L}}$ to denote its set of programs, and we extend this convention to other fragments presented in this work.
We write $\neg \varphi$ as a shorthand of $\varphi \to \bot$, and $\top$ as an abbreviation of $\neg \bot$. The conjunction $\bigwedge \Gamma$ of a finite set of formulas $\Gamma$ is recursively defined as $\bigwedge \emptyset := \top$ and $\bigwedge (\{\varphi\} \cup \Gamma) := \varphi \wedge \bigwedge \Gamma$. 

To deal with programs interpretation, we use standard relation-algebraic notation: for a binary relation \(R\), we write $R^+$ to denote its transitive closure, \(R^*\) its reflexive-transitive closure, \((R;S)\) its composition with the relation $S$ (i.e. $x(R;S)y \iff \exists z. \ x R z \land z S y$), and $R \cup S$ for the union (i.e. $x (R \cup S) y \iff x R y \vee x S y$). The \(n\)-fold composition \(R^n\) is defined by \(R^0 := \mathrm{Id}\) the identity relation, and \(R^{n+1} := (R;R^n)\). The union $\bigcup \Lambda$ of a set of programs $\Lambda$ is recursively defined as $\bigcup \emptyset := \bot?$ and $\bigcup (\{\alpha\} \cup \Lambda) := \alpha \cup \bigcup \Lambda$.

Although our syntax is slightly redundant over the classical setting (e.g. $\langle\alpha\rangle\varphi$ is notation for $\neg[\alpha]\neg\varphi$, and disjunction can be defined in terms of implication and negation), we retain a common language for both the classical and constructive logics to facilitate readability and presentation. The relevant fragments of the language considered in this work are:
\begin{itemize}
    \item $\mathcal L^\Box$ is the diamond-free fragment of $\mathcal{L}$,
    \item $\mathcal{L}^*$ is the $a, a^*$-fragment of $\mathcal{L}$ whose only programs are $a \in \Pi_0$ and $a^*$,
    \item $\mathcal{L}^{\Box, *}$ is the diamond-free fragment of $\mathcal{L}^*$.
\end{itemize}

We now recall some standard definitions over the mutually recursive language of $\sf PDL$ that will be used throughout the paper. First, we define the size of formulas and programs simultaneously.

\begin{definition}
The size of a formula $\varphi \in \mathcal{L}$ and the size of a program $\alpha \in \Pi_\mathcal{L}$, denoted by $|\varphi|$ and $|\alpha|$ respectively, are mutually defined as follows
\begin{description}
    \item $|\bot| := 1$, \quad $|p| := 1$ for $p \in {\rm Prop}$,

    \item $|\varphi \odot \psi| := |\varphi| + |\psi| + 1$ for $\odot \in \{\wedge, \vee, \to\}$,

    \item $|[\alpha] \varphi| := |\alpha| + |\varphi| + 1$,
    
    \item $|\langle \alpha \rangle \varphi| := |\alpha| + |\varphi| + 1$,

    \item $|a| := 1$ for $a \in \Pi_0$,

    \item $|\alpha \odot \beta| := |\alpha| + |\beta| + 1$ for $\odot \in \{; , \cup\}$,

    \item $|\alpha^*| := |\alpha| + 1$,

    \item $|\varphi?| := |\varphi| + 1$.
\end{description}
\end{definition}

The set of subformulas is also presented as a mutual inductive definition over $\mathcal{L}$.

\begin{definition}
The sets of subformulas of a formula $\varphi \in \mathcal{L}$ and a program $\alpha \in \Pi_\mathcal{L}$, denoted by ${\rm Subf}(\varphi)$ and ${\rm Subf}(\alpha)$ respectively, are defined as
\begin{description}
    \item ${\rm Subf}(\top) := \{\top\}$, \quad ${\rm Subf}(p) := \{p\}$ for $p \in {\rm Prop}$,

    \item ${\rm Subf}(\varphi \odot \psi) := \{\varphi \odot \psi\} \cup {\rm Subf}(\varphi) \cup {\rm Subf}(\psi)$ for $\odot \in \{\wedge, \vee, \to\}$,

    \item ${\rm Subf}([\alpha] \varphi) := \{[\alpha] \varphi\} \cup {\rm Subf}(\varphi)$,

    \item ${\rm Subf}(\langle \alpha \rangle \varphi) := \{\langle \alpha \rangle \varphi\} \cup {\rm Subf}(\varphi)$,

    \item ${\rm Subf}(a) := \emptyset$ for $a \in \Pi_0$,

    \item ${\rm Subf}(\alpha \odot \beta) := {\rm Subf}(\alpha) \cup {\rm Subf}(\beta)$ for $\odot \in \{;,\cup\}$,

    \item ${\rm Subf}(\alpha^*) := {\rm Subf}(\alpha)$,

    \item ${\rm Subf}(\psi?) := {\rm Subf}(\psi)$.
\end{description}
\end{definition}

Observe that the set of subformulas includes no programs. Likewise, we also define the set of subprograms over the language $\mathcal{L}$.

\begin{definition}
The sets of subprograms of a formula $\varphi \in \mathcal{L}$ and a program $\alpha \in \Pi_\mathcal{L}$, denoted by ${\rm Subp}(\varphi)$ and ${\rm Subp}(\alpha)$ respectively, are defined as
\begin{description}
    \item ${\rm Subp}(\top) := \emptyset$, \quad ${\rm Subp}(p) := \emptyset$ for $p \in {\rm Prop}$,

    \item ${\rm Subp}(\varphi \odot \psi) := \{\varphi \odot \psi\} \cup {\rm Subp}(\varphi) \cup {\rm Subp}(\psi)$ for $\odot \in \{\wedge, \vee, \to\}$,

    \item ${\rm Subp}([\alpha] \varphi) := {\rm Subp}(\varphi)$,

    \item ${\rm Subp}(\langle \alpha \rangle \varphi) := {\rm Subp}(\varphi)$,

    \item ${\rm Subp}(a) := \{a\}$ for $a \in \Pi_0$,

    \item ${\rm Subp}(\alpha \odot \beta) := {\rm Subp}(\alpha) \cup {\rm Subp}(\beta)$ for $\odot \in \{;,\cup\}$,

    \item ${\rm Subp}(\alpha^*) := \{\alpha^*\} \cup {\rm Subp}(\alpha)$,

    \item ${\rm Subp}(\psi?) := \{\psi?\} \cup {\rm Subp}(\psi)$.
\end{description}
\end{definition}

Similarly, note that the set of programs includes no formulas. Using the previous definition, we define the set of atomic programs of a formula $\varphi \in \mathcal{L}$ as ${\rm At}(\varphi) := {\rm Subp}(\varphi) \cap \Pi_0$.

As one would expect, the number of subformulas and subprograms are bouneded by the corresponding sizes.

\begin{lemma}
\label{AtSize}
Let $\varphi \in \mathcal{L}$ and $\alpha \in \Pi_\mathcal{L}$ be a program.
\begin{enumerate}
    \item $|{\rm Subp}(\alpha)| \leq |\alpha|$,

    \item $|{\rm Subf}(\varphi)| \leq |\varphi|$.
\end{enumerate}
\end{lemma}

\begin{proofidea}
Both statements are shown simultaneously by an easy mutual induction.
\end{proofidea}

We now recall the definition of a (classical) $\sf PDL$-model.

\begin{definition}
A \emph{$\sf PDL$-model} is a triple $\mathfrak{M} = \tuple{W^\mathfrak{M}, \rho^\mathfrak{M}, \|\cdot\|^\mathfrak{M}}$ where $\|\cdot\|^\mathfrak{M} : \mathrm{Prop} \to \mathcal{P}(W^\mathfrak{M})$ is a valuation of the propositional variables and $\rho^\mathfrak{M} : \Pi_0 \to 2^{W^\mathfrak{M} \times W^\mathfrak{M}}$ is an interpretation of atomic programs.
\end{definition}

When clear from context, we may drop superscripts and write e.g.~$\rho$ instead of $\rho^\mathfrak{M}$. Likewise, we may abuse notation and write $w \in \mathfrak{M}$ instead of $w \in W^\mathfrak{M}$. We extend these conventions to the rest of the models presented in this work.

The interpretation of programs $\rho^\mathfrak{M}$ and the valuation $\|\cdot\|^\mathfrak{M}$ can be simultaneously extended all programs and formulas as
\begin{enumerate}[label=\textbullet, wide, labelwidth=!, labelindent=0pt, itemsep=0pt]
    \item $\rho^\mathfrak{M}(\alpha ; \beta) := \rho^\mathfrak{M}(\alpha) ; \rho^\mathfrak{M}(\beta)$,

    \item $\rho^\mathfrak{M}(\alpha \cup \beta) := \rho^\mathfrak{M}(\alpha) \cup \rho^\mathfrak{M}(\beta)$,

    \item $\rho^\mathfrak{M}(\alpha^*) := \big( \rho^\mathfrak{M}(\alpha)\big)^*$,

    \item $\rho^\mathfrak{M}(\varphi?) := \big\{(w,w) \mid w \in \|\varphi\|^\mathfrak{M} \big\}$,

    \item $\|\bot\|^\mathfrak{M} := \emptyset$,

    \item $\|\varphi \wedge \psi\|^\mathfrak{M} := \|\varphi\|^\mathfrak{M} \cap \|\psi\|^\mathfrak{M}$,

    \item $\|\varphi \vee \psi\|^\mathfrak{M} := \|\varphi\|^\mathfrak{M} \cup \|\psi\|^\mathfrak{M}$,

    \item $\|\varphi \to \psi\|^\mathfrak{M} := \big\{w \mid w \in \|\varphi\|^\mathfrak{M} \to w \in \|\psi\|^\mathfrak{M}\big\}$,

    \item $\|[\alpha] \varphi\| := \big\{w \mid \forall v.\  w \rho^\mathfrak{M}(\alpha) v \to v \in \|\varphi\|^\mathfrak{M}\big\}$,

    \item $\|\langle \alpha \rangle \varphi\| := \big\{w \mid \exists v.\  w \rho^\mathfrak{M}(\alpha) v \wedge v \in \|\varphi\|^\mathfrak{M}\big\}$.
\end{enumerate}

Observe that modalities are interpreted in the standard way with respect to the relation $\rho^\mathfrak{M}(\alpha)$. 
We write $\mathfrak{M}, w \Vdash \varphi$ to mean $w \in \|\varphi\|^\mathfrak{M}$, say that a formula $\varphi$ is \emph{valid} in $\mathfrak{M}$ if $\|\varphi\|^\mathfrak{M} = W^\mathfrak{M}$, and $\Vdash_{\sf PDL} \varphi$ indicates that $\varphi$ is valid on all ${\sf PDL}$-models.

We also consider the basic fragment \({\sf K}^*\), which we embed into the constructive \({\sf PDL}\) to establish lower bounds on the complexity of its validity problem. 

\begin{definition}
The basic logic $\sf K^*$ is the fragment of $\sf PDL$ over the language $\mathcal{L}^*$.
\end{definition}

A model $\mathfrak{M}$ is a $\sf K^*$-model if it is a $\sf PDL$-model over $\mathcal{L}^*$, whose only atomic program is $a$. 
As for previous notation, we write $\Vdash_{\sf K^*} \varphi$ to indicate that $\varphi \in \mathcal{L}^*$ is valid on all ${\sf K^*}$-models.

The following theorem summarizes well-known results of $\sf PDL$ and $\sf K^*$, which are considered folklore and form the foundation for our main complexity results.

\begin{theorem}[Fischer and Ladner \cite{FISCHER1979194}, Pratt \cite{PRATT1980231}]\label{ValidityPDL}
\leavevmode
\begin{enumerate}
    \item $\sf PDL$ is decidable in {\sc ExpTime} and has the exponential-size model property.
    \item $\sf K^*$ is {\sc ExpTime}-hard.
\end{enumerate}
\end{theorem}

As a corollary, both $\sf PDL$ and $\sf K^*$ are {\sc ExpTime}-complete.

\section{Constructive {\sf PDL}}\label{secCWMM}

In this section we introduce the constructive propositional dynamic logic of regular programs $\sf CPDL$ and its Wijesekera variant $\sf WPDL$, defined over the language $\mathcal{L}$. The semantic presentation is based on the standard bi-relational semantics of constructive modal logics~\cite{wijesekera1990constructive,mendler2005constructive}.

\begin{definition}\label{def::csf-model}
    A {\em $\sf CPDL$-frame} is a tuple $\mathcal F=\tuple{\domai \Frame ,\peq^\Frame, \rho^\Frame}$, where
    \begin{itemize}
        \item $\domai{\Frame}$ is the set of {\em possible worlds},
        
        \item the {\em intuitionistic relation} $\peq^{\Frame}$ is a preorder,

        \item the function $\rho^\mathcal{F} : \Pi_0 \to 2^{W^\mathcal{F} \times W^\mathcal{F}}$ interprets atomic programs.
\end{itemize}

A  {\em ${\sf CPDL}$-model} is a tuple $\Model=\tuple{\domai{\Model}, \peq^\Model, \rho^\Model, \val\cdot^{\Model}}$ consisting of a $\sf CPDL$-frame equipped with a {\em $\sf CPDL$-valuation} $\val\cdot^{\Model}:\mathrm{Prop} \cup \{\bot\}\to\mathcal P(\domai{\Model})$ such that, for any $p\in\mathrm{Prop}$,
\begin{enumerate}
 
 \item
 $\val\bot^\Model \subseteq\val p^\Model$ \quad 
 \textit{(atomic ex falso quodlibet),}
 
 \item if $w\peq^\mathcal{M} v$ and $w\in \val p^{\Model}$ then $v\in \val p^{\Model}$ \quad \textit{(atomic persistence),}

  \item if $w\in \val\bot^\Model $ and $w\peq^\mathcal{M} v$ or $w \rho^\mathcal{M} (a) v$ for $a \in \Pi_0$, then $v\in \val\bot^\Model$ \quad \textit{(falsum persistence),}

 \item\label{itCondSerialBox} if $w\in \val\bot^\Model $ then there exists $v$ such that $w \rho^\mathcal{M}(a) v$ for any $a \in \Pi_0$ \quad \textit{(falsum seriality).}

\end{enumerate}

A model $\Model$ is {\em infallible} if $\val\bot^\Model=\varnothing$.
A \emph{${\sf WPDL}$-model} is an infallible ${\sf CPDL}$-model.
\end{definition}

\begin{remark}
    To avoid confusion, we reserve $\mathfrak M$ for classical models and $\mathcal M$ for constructive ones. 
\end{remark}

Henceforth, in the figures we present relation $\peq$ with a vertical red arrow, while relations of the kind $\rho(a)$ for $a \in \Pi_0$ are shown in a blue horizontal one.
Solid arrows should be read universally and dashed arrows existentially.

\begin{remark}
Here, we distinguish between the constructive logic $\sf CPDL$ in the style of \cite{Bellin2001ExtendedCC} and the Wijesekera variant $\sf WPDL$ given as in \cite{wijesekera1990constructive}, which differ only in that the former allows $\bot$ to be true on some possible worlds.
Thus, $\sf CPDL$-frames may equivalently be called $\sf WPDL$-frames, but $\sf WPDL$ models validate $\neg\langle \alpha \rangle \bot$, which is not a theorem of $\sf CPDL$.
In Section~\ref{secCW}, we will show that working with Wijesekera models does not lead to any loss in generality when proving our main results.
\end{remark}

Rather than defining satisfiability directly in terms of the interpretation of atomic programs, we introduce a new lawyer to that interpretation using relations $R^\mathcal{M}_\Box : \Pi_0 \to 2^{W^\mathcal{F} \times W^\mathcal{F}}$ and $R^\mathcal{M}_\Diamond : \Pi_0 \to 2^{W^\mathcal{F} \times W^\mathcal{F}}$, which modify the interpretation of $a \in \Pi_0$ depending on whether it occurs within a box or a diamond modality, respectively:
\begin{description}
    \item $R^\mathcal{M}_\Box(a) := {\peq ; \rho(a)}$,

    \item $R^\mathcal{M}_\Diamond(a) := \rho(a)$.
\end{description}

By defining $R^\mathcal{M}_\Box(a)$ in this way, we enforce the Nishimura condition   $R^\mathcal{M}_\Box(a) = {\peq ; R^\mathcal{M}_\Box(a)}$~\cite{NishimuraConstructivePDL}, which must also hold for complex programs.
Here, we follow Nishimura and enforce it explicitly only for `test', as in the other cases it is inherited inductively, with the caveat that the computation of the transitive, reflexive closure must begin with $\peq$ and not the identity, which amounts to adding ${\peq } $ to the standard transitive closure.

To make this precise, we recursively extend the interpretations of programs with respect to relations $R^\mathcal{M}_\Box$ and $R^\mathcal{M}_\Diamond$, while simultaneously extending $\|\cdot\|^\mathcal{M}$ to all formulas as follows, where $\circ \in \{\Box, \Diamond\}$,
\begin{itemize}

    \item $R^\mathcal{M}_\circ (\alpha ; \beta) := R^\mathcal{M}_\circ(\alpha) \ ; R^\mathcal{M}_\circ(\beta)$, 
    
    \item $R^\mathcal{M}_\circ(\alpha \cup \beta) := R^\mathcal{M}_\circ(\alpha) \cup R^\mathcal{M}_\circ(\beta)$, 

    \item $R^\mathcal{M}_\Box (\alpha^*) := {\peq} \cup \big(R_\Box^\mathcal{M}(\alpha)\big)^*$,

    \item $R^\mathcal{M}_\Diamond(\alpha^*) := \big(R_\Diamond^\mathcal{M}(\alpha)\big)^*$,

    \item $R^\mathcal{M}_\Box(\varphi?) := \big\{(w, v) \mid w \peq v \ \wedge \ v \in \|\varphi\|^\mathcal{M} \big\}$ for $\circ \in \{\Box, \Diamond\}$,

    \item $R^\mathcal{M}_\Diamond(\varphi?) := \big\{(w, w) \mid w \in \|\varphi\|^\mathcal{M} \big\}$ for $\circ \in \{\Box, \Diamond\}$,

    \item $\|\varphi \wedge \psi\|^\mathcal{M} := \|\varphi\|^\mathcal{M} \cap \|\psi\|^\mathcal{M}$,

    \item $\|\varphi \vee \psi\|^\mathcal{M} := \|\varphi\|^\mathcal{M} \cap \|\psi\|^\mathcal{M}$,

    \item $\|\varphi \to \psi\|^\mathcal{M} := \big\{ w \mid \forall v \seq w. \ v \in \|\varphi\|^\mathcal{M} \to v \in \|\varphi\|^\mathcal{M} \big\}$,

    \item $\|[\alpha] \varphi\|^\mathcal{M} := \big\{ w \mid \forall v \seq w \ \forall u.  \ v R^\mathcal{M}_\Box(\alpha) u \to u \in \|\varphi\|^\mathcal{M} \big\}$,

    \item $\|\langle \alpha \rangle \varphi\|^\mathcal{M} := \big\{ w  \mid \forall v \seq w \ \exists u. \ v R^\mathcal{M}_\Diamond(\alpha) u  \wedge u \in \|\varphi\|^\mathcal{M} \big\}$,
\end{itemize}

\begin{figure}[h]
    \centering

\begin{tikzpicture}
	\begin{pgfonlayer}{nodelayer}
		\node [style=root, label=right:$\varphi$] (2) at (-2, 2) {$u$};
		\node [style=sroot] (3) at (-3, 2) {};
		\node [style=sroot] (4) at (-3, 1) {};
		\node [style=sroot] (5) at (-4, 1) {};
		\node [style=root, label={left:$[a^*] \varphi$}] (6) at (-4, 0) {$w$};
		\node [style=root, label=left:$\langle a^* \rangle \varphi$] (7) at (1, 0) {$w$};
		\node [style=sroot] (8) at (1, 1) {};
		\node [style=sroot] (12) at (2, 1) {};
		\node [style=root, label=right:$\varphi$] (13) at (3, 1) {$u$};
	\end{pgfonlayer}
	\begin{pgfonlayer}{edgelayer}
		\draw [style=edgeto, red] (6) to (5);
		\draw [style=edgeto, blue] (5) to (4);
		\draw [style=edgeto, red] (4) to (3);
		\draw [style=edgeto, blue] (3) to (2);
		\draw [style=edgeto, red] (7) to (8);
		\draw [style=edgeto, blue, dashed] (8) to (12);
		\draw [style=edgeto, blue, dashed]  (12) to (13);
	\end{pgfonlayer}
\end{tikzpicture}

    \caption{Interpretation of atomic program iteration.}
    \label{fig:placeholder}
\end{figure}

We write $\mathcal{M}, w \Vdash \varphi$ to mean $w \in \|\varphi\|^\mathcal{M}$. A formula $\varphi$ is  \emph{valid} in $\mathcal{M}$ if  $\|\varphi\|^\mathcal{M} = W^\mathcal{M}$, denoted by $\mathcal{M} \Vdash \varphi$. We say that $\varphi$ is \emph{falsifiable} in $\mathcal{M}$ if there is some $w \in W^\mathcal{M}$ such that $\mathcal{M}, w \not\Vdash \varphi$. We write $\varphi \equiv_{\sf CPDL} \psi$ to mean that $\mathcal{M}, w \Vdash \varphi$ if and only if $\mathcal{M}, w \Vdash \psi$ for any $w \in W^\mathcal{M}$ and any $\sf CPDL$-model $\mathcal{M}$.

As is standard, the persistence conditions over the bi-relational models ensures that truth is preserved along intuitionistic paths.

\begin{lemma}[Truth persistence] 
\label{TruthPersistence}
For any $\varphi\in\mathcal L$, any $\sf CPDL$-model $\mathcal M$, and any world $w \in W^{\mathcal M}$, if $\mathcal{M}, w \Vdash \varphi$, then $\mathcal{M}, v \Vdash \varphi$ for any $v \seq w$.
\end{lemma}

\begin{proof}
We proceed by an easy induction on $\varphi \in \mathcal{L}$. The basic cases of $\bot$ and $p \in \mathrm{Prop}$ follow by falsum and atomic persistence, respectively. The remaining cases hold by definition and the inductive hypotheses; we only provide details on the modality formulas and leave the rest to be easily verified by the reader.
\begin{enumerate}[label=\textbullet, wide, labelwidth=!, labelindent=0pt, itemsep=0pt]
    \item[$(\lbrack \alpha \rbrack \varphi$)]: Assume that $\mathcal{M}, w \Vdash [\alpha] \varphi$. In order to show $\mathcal{M}, v \Vdash [\alpha] \varphi$, consider $u \seq v$ and $u R_\Box(\alpha) u'$. By transitivity of the intuitionistic relation we get $u \seq w$. Thus, $\mathcal{M}, u' \Vdash \varphi$ follows by the assumption.

    \item[$\langle \alpha \rangle \varphi$)]: Now consider $\mathcal{M}, w \Vdash \langle \alpha \rangle \varphi$. To conclude with $\mathcal{M}, v \Vdash \langle \alpha \rangle$, consider $u \seq v$. By transitivity of the intuitionistic relation, we get $u \seq w$. Therefore, by the assumption there is some $u'$ such that $u R_\Diamond(\alpha) u'$ and $\mathcal{M}, u' \Vdash \varphi$, concluding the proof. \qedhere
\end{enumerate}
\end{proof}

\begin{remark}
Observe that we did not use Nishimura's condition for proving truth persistence, and indeed persistence holds for arbitary readings of $R^\mathcal{M}_\Box$.
This condition will come into play later for proving Lemma~\ref{AlternativeBox}, which allows us to simplify the semantics of $[\alpha]$ and thus establish soundness of the standard axioms of $\sf CPDL$ (see Corollary~\ref{FormulaEqDiamondCPDL}).
\end{remark}

Therefore, the constructive variants of $\sf PDL$ can be semantically defined as follows.

\begin{definition}
The logic ${\sf CPDL}$ consists of the formulas of $\mathcal{L}$ that are valid on all ${\sf CPDL}$-models. The logic ${\sf WPDL}$ consists of those that are valid on all ${\sf WPDL}$-models.
\end{definition}

We write $\Vdash_{\sf CPDL} \varphi$ (respectively, $\Vdash_{\sf WPDL} \varphi$) to mean that $\varphi \in \mathcal{L}$ is valid on all ${\sf CPDL}$-models (${\sf WPDL}$-models, respectively).

For the constructive {\sf CPDL} (and therefore also {\sf WPDL}) several formulae equivalences that hold in the classical setting can be readily established by unfolding the definitions and appealing to truth persistence (Lemma \ref{TruthPersistence}).

\begin{example}\label{ExFormulaEqDiamondCPDL}
Let $\varphi, \psi \in \mathcal{L}$,  $\equiv$ denote $\equiv_{\sf CPDL}$ and $\Vdash $ denote $\Vdash_{\sf CPDL}$.
\begin{enumerate}

    \item\label{EqCompBox} $\Vdash [ \alpha ] [ \beta ] \varphi \to [ \alpha ; \beta ] \varphi$.

    \item\label{EqUnionBox} $[\alpha \cup \beta] \varphi \equiv [\alpha] \varphi \wedge [\beta] \varphi$.

    \item\label{EqStarBox} $\Vdash \varphi \wedge [\alpha] [\alpha^*] \varphi \to [\alpha^*] \varphi$.

    \item\label{EqTestBox} $[\varphi?] \psi \equiv \varphi \to \psi$.
    
    \item\label{EqCompDiamond} $\Vdash \langle \alpha \rangle \langle \beta \rangle \varphi \to \langle \alpha ; \beta \rangle \varphi$.

    \item\label{EqUnionDiamond} $\langle \alpha \cup \beta \rangle \varphi \equiv \langle \alpha \rangle \varphi \vee \langle \beta \rangle \varphi$.

    \item\label{EqStarDiamond} $\Vdash \varphi \vee \langle \alpha \rangle \langle \alpha^* \rangle \varphi \to \langle \alpha^* \rangle \varphi$.
    
    \item\label{EqTestDiamond} $\langle \varphi? \rangle \psi \equiv \varphi \wedge \psi$.
\end{enumerate}
\end{example}

Next we show that, as promised, the Nishimura condition holds for arbitrary programs, even when we have not enforced it explicitly.

\begin{lemma}
\label{PermutIntRelBox}
${\peq ; R_\Box(\alpha)} = R_\Box(\alpha) $ for any program $\alpha \in \Pi_\mathcal{L}$.
\end{lemma}

\begin{proof}
It is clear that $R_\Box(\alpha) \subseteq {\peq ; R_\Box(\alpha)}  $ since $\peq$ is transitive, so it remains to show that ${\peq ; R_\Box(\alpha)} \subseteq R_\Box(\alpha) $, which follows by a straightforward induction on $\alpha$. We detail the cases for atomic programs and tests, leaving the remaining cases to the reader, as they follow directly from the definitions and the inductive hypotheses.
\begin{enumerate}[label=\textbullet, wide, labelwidth=!, labelindent=0pt, itemsep=0pt]
\item[($a$)]: Consider $w \peq ; R_\Box(a) v$, that is, $w \peq ; \peq ; \rho(a) v$. We obtain $w \peq ; \rho(a)   v$ by transitivity  of $\peq$, which is equivalent to $w R_\Box(a)  v$ by definition.

\item[($\varphi?$)]: Assume $w \peq u R_\Box(\varphi?) v$, that is, $w \peq u \peq v$ and $v \in \|\varphi\|$. By transitivity of $\peq$ we get $w \peq v$, so, by definition, $w R_\Box(\varphi?) v$.  \qedhere
    
\end{enumerate}
\end{proof}

The Nishimura condition will greatly simplify working with boxes, since it allows us to simplify the semantic clause for $[\alpha]$.
This alternative characterization of the box is used throughout the work for technical developments.

\begin{lemma}
\label{AlternativeBox}
Let $\mathcal{M}$ be a $\sf CPDL$-model, $w \in \mathcal{M}$, $\varphi \in \mathcal{L}$ and $\alpha \in \Pi_\mathcal{L}$. Then, 
\[
\mathcal{M}, w \Vdash [\alpha] \varphi \ \Longleftrightarrow \ \forall v. \ w R_\Box(\alpha) v \to \mathcal{M}, v \Vdash \varphi.
\]
\end{lemma}

\begin{proof}
Immediate since $R_\Box(\alpha) = { \peq ; R_\Box(\alpha)} $ by
Lemma~\ref{PermutIntRelBox}.
\end{proof}

With this alternative interpretation in hand, we can complete the equivalences of the $[\alpha]$-formulas in the example.

\begin{corollary}\label{FormulaEqDiamondCPDL}
Let $\varphi, \psi \in \mathcal{L}$.
\begin{enumerate}

    \item\label{EqCompBox} $[ \alpha ; \beta ] \varphi \equiv [ \alpha ] [ \beta ] \varphi$.

    \item\label{EqStarBox} $[\alpha^*] \varphi \equiv \varphi \wedge [\alpha] [\alpha^*] \varphi$.
    
\end{enumerate}
\end{corollary}

\begin{proofidea}
The right to left implication follows straightforwardly given the interpretation of the box modalities in Lemma \ref{AlternativeBox}. \qedhere
\end{proofidea}

\subsection{The diamond-free fragment of ${\sf CPDL}$}

The $\langle\alpha\rangle$-free fragment of the constructive variants of $\sf PDL$ is of particular interest, both as it constitutes the target fragment of the inverse G\"odel–Tarski translation developed in this work, and because it has been studied in the literature \cite{NishimuraConstructivePDL} to which we connect in the subsequent Section \ref{Sec:Nishimura}.

\begin{definition}
The logic ${\sf CPDL}_\Box$ consists of all formulas $\varphi \in \mathcal{L}^\Box$ that are valid on all $\sf CPDL$-models.
\end{definition}

For $\langle \alpha \rangle$-free formulas, validity is independent of model fallibility, so we do not need to distinguish ${\sf CPDL}_\Box$ from ${\sf WPDL}_\Box$, and can always assume infalliblity.
It suffices to observe that for this fragment, satisfiability in a $\sf CPDL$-model is preserved when restricting to the infallible worlds. In cases where an infallible world access a fallible one, \textit{ex-falso} ensures that any satisfiability constraint is trivially satisfied.

\begin{proposition}
\label{EqCKWKBox}
$\Vdash_{\sf CPDL} \varphi$ iff $\Vdash_{\sf WPDL} \varphi$ for any $\varphi \in \mathcal{L}^\Box$.
\end{proposition}

\begin{proof}
We will prove an equivalent statement: 
$\varphi$ is falsifiable in a $\sf CPDL$-model iff $\varphi$ is falsifiable in a $\sf WPDL$-model. The right to left implication is trivial and holds in general since any $\sf WPDL$ model is in particular a $\sf CPDL$ model. We therefore focus on the inverse direction, where we do need to assume that $\varphi \in \mathcal{L}^\Box$. 

Let $\Model$ be a fixed $\sf CPDL$-model satisfying $\Model , w \nVdash \varphi$ for some $w \in \Model$. In particular, $w$ is not a fallible world of $\mathcal{M}$ as that would contradict {\em ex-falso}. We have to exhibit an infallible model $\Model'$ so that for some $w' \in \Model'$ we have $\Model' , w' \nVdash \varphi$ too. To this end, we define $\domai{\Model'}:= \domai{\Model}\setminus \|\bot\|^\mathcal{M}$ 
and define $\peq^{\Model'}, \rho^{\Model'}$
and $\val\cdot^{\Model'}$ to be the restriction of $\peq^{\Model}, \rho^{\Model}$ and $\val\cdot^{\Model}$, respectively, to the set of infallible worlds $W^{\Model'}$. It suffices to show that $\mathcal{M}, v \Vdash \psi$ iff $\mathcal{M}', v \Vdash \psi$ for any $\psi \in \mathcal{L}^\Box$ and any $v\in \Model'$, so we can conclude $\Model ',w\nVdash \varphi$. We proceed by induction on $\psi \in \mathcal{L}^\Box$, providing details for the right to left implication of the $\to$ and $[\alpha]$ cases and leaving the remaining ones to be easily verified by the reader.
\begin{enumerate}[wide, labelwidth=!, labelindent=0pt, itemsep=0pt]
    \item[($\eta \to \phi$)]: Assuming $\mathcal{M}', v \Vdash \eta \to \phi$, we take $u \seq^\mathcal{M} v$ such that $\mathcal{M}, u \Vdash \eta$. If $u \in \|\bot\|^\mathcal{M}$ we trivially conclude that $\mathcal{M}, u \Vdash \phi$. Otherwise, then $u \in \mathcal{M'}$ so by the inductive hypothesis we have $\mathcal{M}', u \Vdash \eta$. Therefore, by assumption we deduce $\mathcal{M}', u \Vdash \phi$ and again by the inductive hypothesis, $\mathcal{M}, u \Vdash \phi$.

    \item[($\lbrack \alpha\rbrack \psi$)]: Assume $\mathcal{M}', v \Vdash [\alpha] \psi$ and take $u \seq v$ and $u' \in\mathcal{M}$ such that $u R_\Box^\mathcal{M}(\alpha) u'$. If also $u \seq^\mathcal{M'} v$ and $u R^\mathcal{M}_\Box(\alpha) u'$, by assumption we have that $\mathcal{M}', u' \Vdash \psi$ so we conclude $\mathcal{M}, u' \Vdash \psi$ by inductive hypothesis. Otherwise, that would mean that either $v \in \|\bot\|^\mathcal{M}$ or there is some fallible world in the chain of $u R^\mathcal{M}_\Box(\alpha) u'$. In both scenarios it is easy to see by falsum persistency that $u' \in \|\bot\|^\mathcal{M}$ so we can trivially conclude that $\mathcal{M}, u' \Vdash \psi$.  \qedhere
\end{enumerate}
\end{proof}

Therefore, we simply write $\Vdash_{\sf CPDL_\Box} \varphi$ to indicate that $\varphi \in \mathcal{L}^\Box$ is valid on all $\sf CPDL$-models (and $\sf WPDL$-models by the previous statement). 

The $a, a^*$-fragment of $\sf CPDL_\Box$ constitutes the target logic into which the basic fragment of $\sf PDL$ is embedded into.

\begin{definition}
The logic $\sf CK_{\Box}^*$ is the fragment of $\sf CPDL_\Box$ over the language $\mathcal{L}^{\Box, *}$.
\end{definition}

We say that a model $\mathcal{M}$ is a $\sf CK^*_\Box$-model if it is a $\sf CPDL$-model over $\mathcal{L}^{\Box, *}$, and we write $\Vdash_{\sf CK^*_\Box} \varphi$ to say that $\varphi \in \mathcal{L}^{\Box, *}$ is valid in all $\sf CK^*_\Box$-models.

\begin{remark}
Since there is only one atomic program in the fragment $\mathcal{L}^{\Box, *}$, the notion of a $\sf CK^*_\Box$-model corresponds directly to the standard bi-relational semantics of constructive modal logics~\cite{wijesekera1990constructive,mendler2005constructive}. In particular, $\sf CK^*_\Box$ coincides with both the logic $\sf IM_K$ studied by Afshari et al. \cite{AfshariGLZ24}, and the logic $\sf IM^*_\Box$ studied by Celani \cite{celani} if one sees the modalities $\Box$ as $[a]$ and $\Box^*$ as $[a^*]$.
\end{remark}

\begin{remark}
Celani \cite{celani} proposed a sequent calculus for $\sf CK_{\Box}^*$, which he calls $\sf IM^*_\Box$, and proved its completeness with respect to $\sf WK$-models. By Proposition \ref{EqCKWKBox}, it is also complete for the class of $\sf CK$-models.
\end{remark}

\section{Relating to the literature: Nishimura's constructive $\sf PDL$}\label{Sec:Nishimura}

To further demonstrate the adequacy of our approach, we establish the equivalence between Nishimura’s constructive $\sf PDL$ \cite{NishimuraConstructivePDL} and $\sf CPDL_\Box$, in the sense that they validate exactly the same formulas.

First, recall that Nishimura's constructive $\sf PDL$ is based on {\em structures} to interpret the formulas.

\begin{definition}[\cite{NishimuraConstructivePDL}]
A {\em structure} $\mathcal{S}$ is a tuple $(W^\mathcal{S}, \leqq^\mathcal{S}, \rho^\mathcal{S}, \|\cdot\|^\mathcal{S})$ where
\begin{itemize}
    \item $W^\mathcal{S}$ is a non-empty set of worlds,

    \item $\leqq^\mathcal{S}$ is a partial order on $W^\mathcal{S}$,

    \item $\rho^\mathcal{S} : \Pi_0 \to 2^{W^\mathcal{S} \times W^\mathcal{S}}$ assigns to each $a \in \Pi_0$ a binary relation $\rho^\mathcal{S}(a)$ such that if $w \leqq^\mathcal{S} w'$ and $w' \rho^\mathcal{S}(a) v$ then $w \rho^\mathcal{S}(a) v$ for any $w, w', v \in W^\mathcal{S}$,

    \item $\|\cdot \|^\mathcal{S} :\mathrm{Prop} \to\mathcal P(\domai{\mathcal{S}})$ satisfies that if $w \in \|p\|^\mathcal{S}$ and $w \leqq^\mathcal{S} w'$ then $w' \in \|p\|^\mathcal{S}$ for any $p\in\mathrm{Prop}$ and $w, w' \in W^\mathcal{S}$.
\end{itemize}
\end{definition}

Note that the relation $\leqq^\mathcal{S}$ is not merely a preorder but a partial order, i.e. an antisymmetric preorder. 

\begin{remark}\label{NishimuraToCPDL}
Any structure $\mathcal{S}$ can be seen as a $\sf WPDL$-model. It suffices to extend the domain of the valuation with $\|\bot\|^\mathcal{S} = \emptyset$ so that falsum persistence, seriality and atomic ex falso are trivially satisfied.
\end{remark}

The functions $\rho^\mathcal{S}$ and $\| \cdot \|^\mathcal{S}$ are extended to interpret all programs and formulae of $\mathcal{L}^\Box$ as
\begin{itemize}
    \item $\rho^\mathcal{S}(\alpha;\beta) := \rho^\mathcal{S}(\alpha);\rho^\mathcal{S}(\beta)$,

    \item $\rho^\mathcal{S}(\alpha \cup \beta) := \rho^\mathcal{S}(\alpha) \cup \rho^\mathcal{S}(\beta)$,

    \item $\rho^\mathcal{S}(\alpha^*) := \rho^\mathcal{S}(\top?) \cup \big(\rho^\mathcal{S}(\alpha)\big)^+$,

    \item $\rho^\mathcal{S}(\varphi ?) := \big\{(w, v) \mid w \leqq v \ \wedge \ v \in \| \varphi \|^\mathcal{S}\big\}$,

    \item $\|\bot\|^\mathcal{S} := \emptyset$,

    \item $\|\varphi \wedge \psi\|^\mathcal{S} := \|\varphi\|^\mathcal{S} \cap \|\psi\|^\mathcal{S}$,

    \item $\|\varphi \vee \psi\|^\mathcal{S} := \|\varphi\|^\mathcal{S} \cup \|\psi\|^\mathcal{S}$,

    \item $\|\varphi \to \psi\|^\mathcal{S} := \big\{ w \mid \forall v. \ w \leqq v \wedge v \in \|\varphi \|^\mathcal{S} \to v \in \|\psi \|^\mathcal{S} \big\}$,

    \item $\|[\alpha] \varphi\|^\mathcal{S} := \big\{w \mid \forall v.\  w \rho^\mathcal{S}(\alpha) v \to v \in \|\varphi\|^\mathcal{S}\big\}$.
\end{itemize}

We write $\mathcal{S}, w \Vdash \varphi$ to mean $w \in \|\varphi\|^\mathcal{S}$, and we say that a formula $\varphi \in \mathcal{L}^\Box$ is \emph{valid} in the structure $\mathcal{S}$ if  $\|\varphi\|^\mathcal{S} = W^\mathcal{S}$, denoted by $\mathcal{S} \Vdash \varphi$.

The following proposition establishes two key properties of Nishimura's semantics: the persistence of the confluence property over complex programs, and the upward persistence of truth. Both properties are established by induction on the programs and formulas, respectively.

\begin{proposition}[\cite{NishimuraConstructivePDL}]
\label{PropNishimura}
Let $\varphi \in \mathcal{L}^\Box$, $\alpha \in \Pi_{\mathcal{L}^\Box}$, let $\mathcal{S}$ be a structure and let $w, v, u \in \mathcal{S}$. Then,
\begin{enumerate}
    \item[(i)]\label{NishimuraPersistenceRel} if $w \leqq v$ and $v \rho(\alpha) u$ then $w \rho(\alpha) u$,
    
    \item[(ii)]\label{NishimuraPersistence} $w \leqq v$ and $w \in \|\varphi\|^\mathcal{S}$ imply $v \in \|\varphi\|^\mathcal{S}$.
\end{enumerate}
\end{proposition}

Therefore, Nishimura's constructive $\sf PDL$ is defined over the $\langle \alpha \rangle$-free fragment $\mathcal{L}^\Box$ based on the structures presented above. 

\begin{definition}
The logic ${\sf CPDL}^\mathcal{N}_\Box$ consists of formulas of $\mathcal{L}^\Box$ that are valid on all structures.\footnote{For ease of presentation, we retain our syntax although it is straightforward to see that it is equivalent to Nishimura's original language: the constant $\bot$ can be introduced as primitive without affecting expressivity, while negation $\neg\varphi$ can be likewise defined as the abbreviation $\varphi \to \bot$.}
\end{definition}

We now show that Nishimura's logic $\sf CPDL^\mathcal{N}_\Box$ and $\sf CPDL_\Box$ validate the same formulas. Since every structure is a $\sf WPDL$-model, it suffices to prove the inverse: any $\sf WPDL$-model can be transformed into a structure that preserves the truth of all formulas.

\begin{proposition}\label{CPDLToNishimura}
Let $\mathcal{M} = \langle W, \peq, \rho, \|\cdot\|\rangle$ be a $\sf WPDL$-model. Consider the equivalence relation $\sim$ over $W$ defined as $w \sim v$ iff $w \peq v \peq w$.
 Set $\mathcal{S} := \langle W / \sim, \ \leqq, \ \rho^\mathcal{S} , \|\cdot\|^\mathcal{S} \rangle$ where
\begin{itemize}
    \item $[w] \leqq [v]$ iff $w \peq v$,

    \item $[w] \rho^\mathcal{S}(a) [v]$ iff $w' \peq ; \rho(a) v'$ for some $w' \sim w$ and $v' \sim v$,

    \item $[w] \in \|p\|^\mathcal{S}$ iff $w' \in \|p\|$ for some $w' \sim w$.
\end{itemize}

Then, $\mathcal{S}$ is a structure and for any $w \in W$, formula $\varphi \in \mathcal{L}^\Box$ and program $\alpha \in \Pi_{\mathcal{L}^\Box}$,
\begin{enumerate}
    \item\label{NishimuraRel} $w \peq ; R_\Box(\alpha) ; \peq v $ iff $ [w] \rho^\mathcal{S}(\alpha) ; \leqq [v]$, 
    
    \item\label{NishimuraSat} $\mathcal{M}, w \Vdash \varphi$ iff $\mathcal{S}, [w] \Vdash \varphi$.
\end{enumerate}
\end{proposition}

\begin{proof}
First, observe that $\leqq$, $\rho^\mathcal{S}$, and $\|\cdot\|^\mathcal{S}$ are well-defined. In particular, $\leqq$ is well-defined by the transitivity of $\peq$; we omit the standard proof. Moreover, it is easy to see that $\mathcal{S}$ is a structure, whose details are left to the reader and simply follow by definition, atomic persistence and transitivity of $\peq$.

We proceed to show the statements \ref{NishimuraRel} and \ref{NishimuraSat} by mutual induction.
\begin{enumerate}[wide, labelwidth=!, labelindent=0pt, itemsep=5pt]
    \item[\ref{NishimuraRel}.] The case of atomic programs follow directly from the definition together with the transitivity of $\peq$, while the case for $\alpha\cup\beta$ follows immediately from the inductive hypotheses. We provide details for the remaining program cases.
    \begin{enumerate}[wide, labelwidth=!, labelindent=10pt, itemsep=2pt]
        \item[($\alpha ; \beta$)]: First, assume that $w \peq ; R_\Box(\alpha ; \beta) ; \peq v$, which implies $w \peq ; R_\Box(\alpha);  \peq u \peq ; R_\Box(\beta) ; \peq v$ exploiting the reflexivity of $\peq$. By the inductive hypotheses we obtain $[w] \rho^\mathcal{S}(\alpha) ; \leqq [u] \rho^\mathcal{S}(\beta) ; \leqq [v]$. Since $\leqq ; \rho^\mathcal{S}(\beta) \subseteq \rho^\mathcal{S}(\beta)$ after Proposition \ref{PropNishimura}, we conclude that $[w] \rho^\mathcal{S}(\alpha ; \beta) ; \leqq [v]$.

        Now assume $[w] \rho^\mathcal{S}(\alpha ; \beta) ; \leqq [v]$, which implies $[w] \rho^\mathcal{S}(\alpha) ; \leqq [u] \rho^\mathcal{S}(\beta) ; \leqq [v]$ by reflexivity. By the inductive hypothesis we get $w \peq ; R_\Box(\alpha) ; \peq ; R_\Box(\beta) ; \peq v$ and, since we have that $\peq ; R_\Box(\beta) \subseteq R_\Box(\beta) ; \peq$ by Lemma \ref{PermutIntRelBox}, we conclude that $w \peq ; R_\Box(\alpha ; \beta) ; \peq v$.

        \item[($\alpha^*$)]: First assume that $w \peq ; R_\Box(\alpha^*) ; \peq v$, that is, $w \peq ; R_\Box(\alpha)^* ; \peq v$. We proceed by the induction on the number of $R_\Box(\alpha)$-steps. If no step is taken and simply $w \peq v$, that would directly imply $[w] \rho^\mathcal{S}(\top?) [v]$ so we obtain $[w] \rho^\mathcal{S}(\alpha^*) [v]$ by definition and $[w] \rho^\mathcal{S}(\alpha^*) ; \leqq [v]$ by reflexivity. For the inductive step, assume $w \peq ; R_\Box(\alpha)^* ; R_\Box(\alpha) \peq v$, which implies $w \peq ; R_\Box(\alpha)^* ; \peq u \peq ; R_\Box(\alpha) ; \peq v$ by reflexivity. Applying both inductive hypotheses on the number of steps and on the program, we obtain $[w] \rho^\mathcal{S}(\alpha^*) ; \leqq ; \rho^\mathcal{S}(\alpha) ; \leqq [v]$. Since $\leqq ; \rho^\mathcal{S}(\alpha) \subseteq \rho^\mathcal{S}(\alpha)$ by Proposition \ref{PropNishimura}, we deduce $[w] \rho^\mathcal{S}(\alpha^*) ; \rho^\mathcal{S}(\alpha) ; \leqq [v]$ and conclude $[w] \rho^\mathcal{S}(\alpha^*) ; \leqq [v]$ by definition. 

        Now suppose that $[w] \rho^\mathcal{S}(\alpha^*) ; \leqq [v]$, that is, $[w] \rho^\mathcal{S}(\top?) \cup \big(\rho^\mathcal{S}(\alpha)\big)^+ ; \leqq [v]$. If $[w] \rho^\mathcal{S}(\top?) ; \leqq [v]$, by definition we get $w \peq v$ and we trivially conclude $w \peq ; R_\Box(\alpha^*) ; \peq v$. Otherwise, if $[w] \big(\rho^\mathcal{S}(\alpha)\big)^+ ; \leqq [v]$, we proceed by induction on $n - 1$, for $n$ the number of $\rho^\mathcal{S}(\alpha)$-steps taken. For the base case we simply assume $[w] \rho^\mathcal{S}(\alpha) ; \leqq [v]$, so by the inductive hypothesis on the program we get $w \peq ; R_\Box(\alpha) ; \peq v$ which implies $w \peq ; R_\Box(\alpha^*) ; \peq v$. For the inductive step consider $[w] \big(\rho^\mathcal{S}(\alpha)\big)^+ ; \rho^\mathcal{S}(\alpha) ; \leqq [v]$, which implies $[w] \big(\rho^\mathcal{S}(\alpha)\big)^+ ; \leqq ; \rho^\mathcal{S}(\alpha) ; \leqq [v]$ by reflexivity. Then, by both the inductive hypotheses on the number of steps and on the program we obtain $w \peq ; R_\Box(\alpha^*) ; \peq ; R_\Box(\alpha) ; \peq v$. Thus, since we have $\peq ; R_\Box(\alpha) \subseteq R_\Box(\alpha) ; \peq$ by Lemma \ref{PermutIntRelBox}, we get $w \peq ; R_\Box(\alpha^*) ; R_\Box(\alpha) ; \peq v$ and conclude $w \peq ; R_\Box(\alpha^*) ; \peq v$.

        \item[($\varphi?$)]: First assume that $w \peq ; R_\Box(\varphi?) ; \peq v$, i.e. $w \peq u \peq u' \peq v$ and $\mathcal{M}, u' \Vdash \varphi$. By truth persistence (Lemma \ref{TruthPersistence}) we obtain $\mathcal{M}, v \Vdash \varphi$. Therefore, by the statement \ref{NishimuraSat} we obtain $\mathcal{S}, [v] \Vdash \varphi$ so, by definition, $[w] \rho^\mathcal{S}(\varphi?) [v]$. Thus, by reflexivity we conclude that $[w] \rho^\mathcal{S}(\varphi?) ; \leqq [v]$.

        Now consider $[w] \rho^\mathcal{S}(\varphi?) ; \leqq [v]$, that is, $[w] \leqq [u] \leqq [v]$ with $\mathcal{S}, [u] \Vdash \varphi$. By truth persistence of the structures (Proposition \ref{PropNishimura}) we get $\mathcal{S}, [v] \Vdash \varphi$ and, by the statement \ref{NishimuraSat}, $\mathcal{M}, v \Vdash \varphi$. Thus, we have $w R_\Box(\varphi?) v$ by definition, which implies $w \peq ; R_\Box(\varphi?) ; \peq v$ by reflexivity. 
    \end{enumerate}

    \item[\ref{NishimuraSat}.] The formula cases are easily satisfied by definition and the inductive hypothesis, whose details are left to the reader to verify. In particular, for the case of propositional atoms, the left to right implication follows by reflexivity of $\sim$, while the inverse holds by atom persistence. As an illustration, we provide further details for the box formula.
    \begin{enumerate}[wide, labelwidth=!, labelindent=10pt, itemsep=0pt]
        \item[($\lbrack \alpha \rbrack \varphi$)]: First, assuming $\mathcal{M}, w \Vdash [\alpha] \varphi$ we want to show that $\mathcal{S}, [w] \Vdash [\alpha] \varphi$. For this, consider $[w] \rho^\mathcal{S}(\alpha) \ [v]$, which implies $[w] \rho^\mathcal{S}(\alpha) ; \leqq [v]$ by reflexivity. Then, by the statement \ref{NishimuraRel} we obtain $w \peq ; R_\Box(\alpha) \ u \peq v$. Hence, from the assumption we deduce that $\mathcal{M}, u \Vdash \varphi$ and, by truth persistence (Lemma \ref{TruthPersistence}) we get $\mathcal{M}, v \Vdash \varphi$. Thus, by the inductive hypothesis we conclude that $\mathcal{S}, [v] \Vdash \varphi$.

        Finally, assuming $\mathcal{S}, [w] \Vdash [\alpha] \varphi$ we want to show that $\mathcal{M}, w \Vdash [\alpha] \varphi$. For this, consider $w  \peq ; R_\Box(\alpha) v$, which implies $w \peq ; R_\Box(\alpha) ; \peq v$ by reflexivity. By the statement \ref{NishimuraRel} we obtain that $[w] \rho^\mathcal{S}(\alpha) [u] \leqq [v]$. Therefore, from the assumption we deduce $\mathcal{S}, [u] \Vdash \varphi$ which implies $\mathcal{S}, [v] \Vdash \varphi$ by truth persistence (Proposition \ref{PropNishimura}). Thus, by the inductive hypothesis we conclude that $\mathcal{M}, v \Vdash \varphi$. \qedhere
    \end{enumerate}
\end{enumerate}
\end{proof}

Thus, Nishimura's constructive $\sf PDL$ and $\sf CPDL_\Box$ are equivalent in terms of the formulas they validate.

\begin{theorem}\label{EqNishimura}
$\Vdash_{\sf CPDL_\Box^\mathcal{N}} \varphi$ if and only if $\Vdash_{\sf CPDL_\Box} \varphi$ for any $\varphi \in \mathcal{L}^\Box$.
\end{theorem}

\begin{proof}
The right to left implication is immediate since any structure can be seen as a $\sf WPDL$-model. For the inverse implication, assume $\mathcal{M}, w \not\Vdash \varphi$ for $\mathcal{M}$ a $\sf CPDL$-model and $w \in \mathcal{M}$. Given that validity is independent of fallibility over the $\langle \alpha \rangle$-free fragment (Proposition \ref{EqCKWKBox}), there is a $\sf WPDL$-model $\mathcal{M}'$ satisfying $\mathcal{M}', v \not\Vdash \varphi$ for some $v \in \mathcal{M}'$. The result then follows from Proposition \ref{CPDLToNishimura}.
\end{proof}

\section{Translating $\sf CPDL$ into $\sf WPDL$}\label{secCW}

In this section, we show that constructive $\sf CPDL$ embeds into its `infallible' Wijesekera variant $\sf WPDL$.
This will allow us to focus our attention mostly on $\sf WPDL$ in subsequent sections, with results applying to $\sf CPDL$ as well.
The intuition for the translation is simple: since, constructively, $\bot$ is essentially a variable which implies all other variables, we can replace it by the conjunction of said variables over any accessible world. This allows us to simulate the behaviour of fallible worlds within the infallible setting.

There are, however, a couple of caveats to this idea.
The first is that the translation must be parameterised by the formula $\varphi$ being translated, and is therefore denoted by $\omega_\varphi$. This allows us to restrict attention to the finite set ${\rm Var}(\varphi)$ of propositional variables occurring in $\varphi$. Similarly, we restrict attention to the atomic programs in $A := {\rm At}(\varphi)$, which determine the relevant accessibility relations. These restrictions ensure that the translation has polynomial size.

Second, we cannot simply replace $\bot$ with the conjunction of the variables occurring in $\varphi$. Otherwise, the formula $[(\bigcup_{a\in A} a)^*]\bigwedge {\rm Var}(\varphi)\to\bot$ would translate into a valid formula, since its antecedent would coincide with the translation of $\bot$. To avoid this, we introduce a fresh propositional variable $p_\bot$ that does not occur in $\varphi$ and include it in the conjunction.

The third caveat is that the translation must ensure that the formula replacing $\bot$ correctly characterizes the relevant fallible worlds, which would be taken care of by setting $\omega_\varphi(\bot) =[(\bigcup_{a \in A} a)^*] \bigwedge P $ for $P := {\rm Var}(\varphi) \cup \{p_\bot\}$. However, this condition alone does not guarantee that the accessibility relations remain serial when restricted to the worlds satisfying $\omega_\varphi(\bot)$.
This is solved by adding pretty much any conjunct of the form $\langle a \rangle \psi$ for each $a \in A$, and the simplest choice seems to be $\langle a \rangle p_\bot$. Hence, the restriction of each relevant accessibility relation to the worlds satisfying $\omega_\varphi(\bot)$ remains serial.

Finally, although the guiding idea is to replace $\bot$, we define $\omega_\varphi$ by mutual recursion on formulas and programs in accordance with the mutually recursive syntax of the language. This ensures that the translation is well-defined on both syntactic categories and facilitates the inductive proofs that follow. We thus arrive at the following.

\begin{definition}
Let $\varphi \in \mathcal L$ formula and  $p_\bot$ be a variable not appearing in $\varphi$. The translations of a formula and a program over $\omega_\varphi$ are mutually defined as
\begin{itemize}
    \item $\omega_\varphi (\bot) := \big[(\bigcup\limits_{a \in A} a)^*] \big(\bigwedge P \wedge \bigwedge\limits_{a \in A} \langle a \rangle p_\bot \big)$ for $P := {\rm Var}(\varphi) \cup \{p_\bot\}$ and $A := {\rm At}(\varphi)$,
    
    \item $\omega_\varphi (p):= p$,

    \item $\omega_\varphi (\psi \odot \phi) := \omega_\varphi (\psi) \odot \omega_\varphi(\phi)$ for $\odot \in \{\wedge, \vee, \to\}$,

    \item $\omega_\varphi([\alpha]\psi) := [\omega_\varphi(\alpha)] \omega_\varphi (\psi)$,

    \item $\omega_\varphi(\langle \alpha\rangle \psi) := \langle \omega_\varphi(\alpha)\rangle  \omega_\varphi (\psi)$,

    \item $\omega_\varphi(a) := a$ if $a \in {\rm At}(\varphi)$ and $\omega_\varphi(a) := \theta$ otherwise, 

    \item $\omega_\varphi(\alpha \circ \beta) := \omega_\varphi(\alpha) \circ \omega_\varphi(\beta)$ for $\circ \in \{;, \cup\}$,

    \item $\omega_\varphi(\alpha^*) := \omega_\varphi(\alpha)^*$,

    \item $\omega_\varphi(\psi?) := \omega_\varphi(\psi)?$.
\end{itemize}
\end{definition}

It is essential that this translation is polynomial.

\begin{proposition}
\label{PolynomialOmega}
$|\omega_\varphi(\varphi)|$ is polynomial on $|\varphi|$.
\end{proposition}

\begin{proofidea}
Given that $\Big|[(\bigcup\limits_{a \in A} a)^*](\bigwedge P \wedge \bigwedge\limits_{a \in A} \langle a \rangle p_\bot)\Big| \peq 7 \cdot |\varphi|^2$ using Lemma \ref{AtSize}, it suffices to show by an easy mutual induction that $|\omega_\varphi(\alpha)| \peq 7 \cdot |\alpha| \cdot |\varphi|^2$ and $|\omega_\varphi(\psi)| \peq 7 \cdot |\psi| \cdot |\varphi|^2$ for any $\alpha \in {\rm Subp}(\varphi)$ and any $\psi \in {\rm Subf}(\varphi)$. 
\end{proofidea}

The intention is for $\omega_\varphi$ to preserve validity, i.e.~that it is  {\em correct} (valid formulas map to valid formulas) and {\em exact} (falsifiable formulas map to falsifiable formulas).
To establish exactness, we construct the infallible model corresponding to a given fallible one by interpreting $p_\bot$ as $\bot$.

\begin{proposition}
\label{ValidityCtoW}
Let $\varphi \in \mathcal{L}$ and let $\Model = \tuple{W,\peq,\rho,\val\cdot }$ be a ${\sf CPDL}$-model.
Define $\Model^{\sf W}_\varphi := \tuple{W,\peq,\rho^{\sf W}_\varphi,\val\cdot^{\sf W}_\varphi}$ by setting 
\[
\begin{aligned}
\val p^{\sf W} &:= 
\begin{cases}
    \val{\bot} &\text{if $p=p_\bot$}\\
    \varnothing &\text{if $p=\bot$}\\
    \val p &\text{otherwise}
\end{cases}
&\qquad \qquad
\rho^{\sf W}_\varphi (a) &:=
\begin{cases}
    \rho(a) &\text{if $a\in{\rm At}(\varphi)$}\\
    {\rm Id} &\text{otherwise.}
\end{cases}
\end{aligned}
\]

Then, $\Model^{\sf W}_\varphi$ is a ${\sf WPDL}$-model and
\begin{enumerate}
    \item\label{ExactnessWPDLCPDLRel} $R_\circ^\mathcal{M}(\alpha) = R_\circ^{\mathcal{M}_\varphi^{\sf W}}(\omega_\varphi(\alpha))$ for any $\alpha \in {\rm Subp}(\varphi)$ and $\circ \in \{\Box, \Diamond\}$,

    \item\label{ExactnessWPDLCPDL} $\Model,w \Vdash \psi$ iff $\Model^{\sf W}_\varphi,w \Vdash \omega_\varphi(\psi)$ for any $w\in W$ and $\psi \in {\rm Subf}(\varphi)$.
\end{enumerate}
\end{proposition}

\begin{proof}
Clearly, $\mathcal{M}^{\sf W}_\varphi$ is a ${\sf WPDL}$-model, whose conditions require straightforward verification and are left to the reader. We proceed by mutual induction for the second part of the proof.
\begin{enumerate}[wide, labelwidth=!, labelindent=0pt, itemsep=3pt]
    \item[\ref{ExactnessWPDLCPDLRel}.] The cases of the programs are trivially satisfied by definition and the inductive hypothesis. In particular, for an atomic program $a \in {\rm Subp}(\varphi)$ we have $a \in {\rm At}(\varphi)$ so the atomic interpretation of the models coincide. We leave the cases to be verified by the reader and provide details on the test program interpreted over the box relation, which is satisfied by the statement \ref{ExactnessWPDLCPDL}:
    \begin{equation*}
        \begin{split}
             R^\mathcal{M}_\Box(\psi?) & = \{ (w,v) \mid w \peq v \ \wedge \ v \in \|\psi\|^\mathcal{M}\}  \\ & = \{ (w,v) \mid w \peq v \ \wedge \ v \in \|\omega_\varphi(\psi)\|^{\mathcal{M}^{\sf W}_\varphi}\} \\ & = R_\Box^{\mathcal{M}^{\sf W}_\varphi}(\omega_\varphi(\psi)?) = R_\Box^{\mathcal{M}^{\sf W}_\varphi}(\omega_\varphi(\psi?)).
        \end{split}
    \end{equation*}

    \item[\ref{ExactnessWPDLCPDL}.] We consider the formula cases. Observe that the case of a propositional variable is trivially satisfied since $p_\bot$ is not a subformula of $\varphi$ by assumption. We provide details for the $\bot$ and the diamond formula, leaving the remaining scenarios to be easily verified by the reader.
    \begin{enumerate}[wide, labelwidth=!, labelindent=3pt, itemsep=3pt]
        \item[($\bot$)]: First assume $\mathcal{M}^{\sf W}_\varphi , w \Vdash \omega_\varphi(\bot)$. Since $p_\bot \in P$ we in particular deduce that $w \in \|p_\bot\|^{\sf W}_\varphi = \|\bot\|$. For the inverse implication, consider $w \in \|\bot\|$ and $w R_\Box\big((\bigcup_{a \in A} a)^*\big) v$. By the alternative interpretation of the box given in Lemma \ref{AlternativeBox}, it suffices to show that $\mathcal{M}^{\sf W}_\varphi, v \Vdash \bigwedge P \wedge \bigwedge_{a \in A}\langle a \rangle p_\bot$. First observe that from $w R_\Box\big((\bigcup_{a \in A} a)^*\big) v$ we can easily deduce $v \in \|\bot\|$ by induction on the size of $A$ and the number of steps in the reflexive transitive closure using falsum persistence, whose details are left to the reader to verify. Therefore, we obtain $\mathcal{M}_\varphi^{\sf W}, v \Vdash p_\bot$ by definition. For the remaining propositional variables in $P$, by the first condition of the $\sf CPDL$-model we have that $\|\bot\| \subseteq \|p\| = \|p\|^{\sf W}_\varphi$, so we show that $\mathcal{M}_\varphi^{\sf W}, v \Vdash \bigwedge P$. 
        We now need to conclude with $\mathcal{M}^{\sf W}_\varphi, v \Vdash \bigwedge_{a \in A}\langle a \rangle p_\bot$.
        If $A = \emptyset$, the result is trivially satisfied. Otherwise, consider $a \in A$ and $u \seq v$. By falsum persistence we deduce that $u \in \|\bot\|$, so again by falsum seriality and persistence there is some $u'$ such that $u \rho(a) u'$ and $u' \in \|\bot\| = \|p_\bot\|^{\sf W}_\varphi$. Thus, we get $\mathcal{M}_\varphi^{\sf W}, v \Vdash \langle a \rangle p_\bot$, concluding the case.

        \item[($\langle \alpha \rangle \varphi$)]: By the inductive hypothesis and the statement \ref{ExactnessWPDLCPDL} we simply conclude that
        \begin{equation*}
            \begin{split}
                \mathcal{M}, w \Vdash \langle \alpha \rangle \psi & \ \Longleftrightarrow \ \forall v \seq w. \ \exists u. \ v R^\mathcal{M}_\Diamond(\alpha) u \ \wedge \ \mathcal{M}, u \Vdash \psi \\ & \ \Longleftrightarrow \ \forall v \seq w. \ \exists u. \ v R^{\mathcal{M}^{\sf W}_\varphi}_\Diamond(\omega_\varphi(\alpha)) u \ \wedge \ \mathcal{M}^{\sf W}_\varphi, u \Vdash \psi \\ & \ \Longleftrightarrow \ \mathcal{M}^{\sf W}_\varphi, w \Vdash \langle \omega_\varphi(\alpha) \rangle  \omega_\varphi(\psi) = \omega_\varphi(\langle \alpha \rangle \psi).
            \end{split}
        \end{equation*}
    \end{enumerate}
\end{enumerate}

\end{proof}

For correctness, we transform an infallible model into a fallible one preserving truth and falsity.
Let $\Model = \tuple{W,\peq,\rho,\val\cdot }$ be a ${\sf WPDL}$-model. For a given $\varphi \in \mathcal{L}$ we define $\Model^{\sf C}_{\varphi} := \tuple{W,\peq,\rho^{\sf C}_\varphi,\val\cdot^{\sf C}_{\varphi} }$ by setting $\|p\|^{\sf C}_{\varphi} := \|p\|$ if $p \in \mathrm{Var}(\varphi)$ and $\|p\|^{\sf C}_{\varphi} := \|\omega_\varphi(\bot)\|$ otherwise; and $\rho^{\sf C}_\varphi(a) = \rho(a)$ if $a \in {\rm At}(\varphi)$ and $\rho^{\sf C}_\varphi(a) = {\rm Id}$ otherwise. Indeed, this construction yields a $\sf CPDL$-model.

\begin{lemma}
\label{IsModelWtoC}
Let $\varphi \in \mathcal{L}$. If $\Model$ is a $\sf WPDL$-model, then 
$\Model^{\sf C}_\varphi$ is a ${\sf CPDL}$-model.
\end{lemma}

\begin{proof}
We verify that $\mathcal{M}^{\sf C}_\varphi$ is a ${\sf CPDL}$-model by checking the model conditions; the frame conditions are preserved by construction.
\begin{enumerate}[wide, labelwidth=!, labelindent=0pt, itemsep=0pt]
\item[1.] First show that $\|\bot\|^{\sf C}_\varphi \subseteq \|p\|^{\sf C}_\varphi$ for any $p \in \mathrm{Prop}$. If $p \notin \mathrm{Var}(\varphi)$, it is trivially satisfied since $\|\bot\|^{\sf CK}_\varphi = \|p\|^{\sf C}_\varphi$. Otherwise, if $p \in \mathrm{Var}(\varphi)$ and $w \in \|\bot\|^{\sf C}_\varphi = \|\omega_\varphi(\bot)\|$, then in particular we have that $w \in \|p\| = \|p\|^{\sf C}_\varphi$, concluding the proof of the condition.

\item[2.] To show atomic persistence, consider $w \peq v$ and $w \in \|p\|^{\sf C}_\varphi$ for $p \in {\rm Prop}$. If $p \in \mathrm{Var}(\varphi)$, then $w \in \|p\|^{\sf C}_\varphi = \|p\|$ and, by the atomic persistence of $\mathcal{M}$, $v \in \|p\| = \|p\|^{\sf C}_\varphi$. Otherwise, if $p \notin \mathrm{Var}(\varphi)$ then $w \in \|p\|^{\sf C}_\varphi = \|\omega_\varphi(\bot)\|$ and, by the truth persistence (Lemma \ref{TruthPersistence}) of $\mathcal{M}$ we conclude $v \in \|\omega_\varphi(\bot)\| = \|p\|^{\sf C}_\varphi$.

 \item[3.] To show falsum persistence, consider $w \in \|\bot\|^{\sf C}_\varphi$ and $w \peq v$ or $w \rho^{\sf C}_\varphi(a) v$ for $a \in \Pi_0$. If $w \peq v$ the result directly follows by the truth persistence (Lemma \ref{TruthPersistence}) of $\mathcal{M}$. Now assume that $w \rho^{\sf C}_\varphi(a) v$. If $a \notin {\rm At}(\varphi)$, then $\rho^{\sf C}_\varphi(a) = {\rm Id}$ and the result is trivially satisfied. Otherwise, if $a \in {\rm At}(\varphi)$, by definition, $\rho^{\sf C}_\varphi(a) = \rho(a)$. Therefore, we have that $w \peq ; \rho(a) v$, which implies $w R_\Box^\mathcal{M}(\bigcup\limits_{a \in A} a) v$. By Corollary \ref{FormulaEqDiamondCPDL} since we have $\mathcal{M}, w \Vdash [(\bigcup\limits_{a \in A} a)^*](\bigwedge P \wedge \bigwedge\limits_{a \in A} \langle a \rangle p_\bot)$ we also get $\mathcal{M}, w \Vdash [\bigcup\limits_{a \in A} a] [(\bigcup\limits_{a \in A} a)^*](\bigwedge P \wedge \bigwedge\limits_{a \in A} \langle a \rangle p_\bot)$. Thus, by the alternative interpretation of the box given in Lemma \ref{AlternativeBox}, since $w R_\Box^\mathcal{M}(\bigcup\limits_{a \in A} a) v$ we conclude that $v \in \big\|[(\bigcup\limits_{a \in A} a)^*](\bigwedge P \wedge \bigwedge\limits_{a \in A} \langle a \rangle p_\bot)\big\| = \|\bot\|^{\sf C}_\varphi$.

\item[4.] Finally, to show falsum seriality consider $w \in \|\bot\|^{\sf C}_\varphi$ and $a \in \Pi_0$. If $a \in {\rm At}(\varphi)$, by assumption and reflexivity we have in particular that $w \in \|\langle a \rangle p_\bot\|$, so in particular there is some $v$ satisfying $w \rho(a) v$, that is, $w \rho^{\sf C}_\varphi(a) v$. Otherwise, if $a \notin {\rm At}(\varphi)$ we can simply conclude with $w \rho^{\sf C}_\varphi(a) w$ since $\rho^{\sf C}_\varphi(a) = {\rm Id}$.  \qedhere
\end{enumerate}   
\end{proof}

The translation is shown to be correct using the model construction described above.

\begin{proposition}
\label{ValidityWtoC}
Let $\Model$ be a $\sf WPDL$-model and let $\varphi \in \mathcal{L}$. Then, 
\begin{enumerate}
    \item\label{CompletenessWPDLCPDLRel} $R_\circ^{\mathcal{M}^{\sf C}_\varphi}(\alpha) = R_\circ^{\mathcal{M}}(\omega_\varphi(\alpha))$ for any $\alpha \in {\rm Subp}(\varphi)$,

    \item\label{CompletenessWPDLCPDL} $\Model^{\sf C}_\varphi,w\Vdash \psi $ iff $\Model,w \Vdash \omega_\varphi(\psi)$ for any $w \in \mathcal{M}$ and $\psi \in {\rm Subf}(\varphi)$.
\end{enumerate}
\end{proposition}

\begin{proof}
We proceed by mutual induction.
\begin{enumerate}[wide, labelwidth=!, labelindent=0pt, itemsep=0pt]
    \item[\ref{CompletenessWPDLCPDLRel}.] The cases of the programs are trivially satisfied by definition and the inductive hypotheses given the subprogram assumption, whose details are left to the reader. As an illustration, we show the case of the test program interpreted over the diamond relation. By definition and the statement \ref{CompletenessWPDLCPDL} we directly conclude that
    \begin{equation*}
    \begin{split}
        R^{\mathcal{M}^{\sf C}_\varphi}_\Diamond(\psi?) & = \{(w,w) \mid w \in \|\psi\|^{\mathcal{M}^{\sf C}_\varphi} \} = \{(w, w) \mid w \in \|\omega_\varphi(\psi)\|^\mathcal{M} \} \\ & = R^\mathcal{M}_\Diamond(\omega_\varphi(\psi)?) = R^\mathcal{M}_\Diamond(\omega_\varphi(\psi?)).        
    \end{split}
    \end{equation*}

    \item[\ref{CompletenessWPDLCPDL}.] The cases of the formulas are easily shown by definition and the inductive hypothesis, whose details are left to the reader for an straightforward verification. For an illustration, we provide details on the box formula, which is satisfied by the inductive hypothesis and the statement \ref{CompletenessWPDLCPDLRel} as
    \begin{equation*}
        \begin{split}
            \mathcal{M}^{\sf C}_\varphi, w \Vdash [\alpha] \psi & \ \Longleftrightarrow \ \forall v \seq w. \ \forall u. \ v R^{\mathcal{M}^{\sf C}_\varphi}_\Box(\alpha) u \to \mathcal{M}^{\sf C}_\varphi, u \Vdash \psi \\ & \ \Longleftrightarrow \ \forall v \seq w. \ \forall u. \ v R^{\mathcal{M}}_\Box(\alpha)(\omega_\varphi(\alpha)) u \to \mathcal{M}, u \Vdash \omega_\varphi(\psi) \\ & \ \Longleftrightarrow \ \mathcal{M}, w \Vdash [\omega_\varphi(\alpha)] \omega_\varphi(\psi) = \omega_\varphi([\alpha] \psi).
        \end{split}
    \end{equation*}
\end{enumerate}
\end{proof}

Thus, the defined translation preserves validity as intended. Therefore, ${\sf CPDL}$ can be seen as a fragment of ${\sf WPDL}$, allowing us to focus on the latter when e.g.~proving decidability results.

\begin{theorem}
\label{CorrespondenceConstructiveAndWijesekera} $\Vdash_{\sf CPDL} \varphi$ iff $\Vdash_{\sf WPDL} \omega_\varphi (\varphi)$ for any $\varphi\in\mathcal L$.
\end{theorem}

\begin{proof}
The right to left implication is due to Proposition \ref{ValidityCtoW}, while the inverse follows from Proposition \ref{ValidityWtoC}.
\end{proof}

\section{Translating $\sf WPDL$  into $\sf PDL$}
\label{sec:WKtoPDL}

Next, we define our G\"odel-Tarski translation from ${\sf WPDL}$ into $\sf PDL$ which, as we have noted, yields a translation of ${\sf CPDL}$ into $\sf PDL$ by pre-composing with the embedding of the previous Section \ref{secCW}.
Recall that in the original G\"odel-Tarski translation of intuitionistic propositional logic into $\sf S4$~\cite{GodelTranslation}, each variable $p$ is replaced by $\Box p$ and each implication $\varphi\to\psi$ is replaced by $\Box(\varphi'\to\psi')$, where $\cdot'$ denotes the recursive application of the translation.

To adapt this approach to $\sf PDL$, we consider a special atomic program $i$ to encode the intuitionistic relation, and assume that it belongs to the language of $\sf PDL$ but not to that of $\sf WPDL$. Therefore, we interpret $\rho(i)$ as the intuitionistic preorder, but as this is not assumed to be transitive and reflexive in the $\sf PDL$ semantics, we use $[i^*]$ instead of simply $[i]$. Accordingly, we use $[i^*]$ in place of $[i]$ in the translation of propositional variables and implications.

Beyond these propositional cases, the translation follows the approach of~\cite{degen}, but the modalities require additional care to ensure upward persistence. To preserve the intended constructive semantics, we mutually define, alongside the formula translation, two translations of each program $\alpha$, namely $\alpha^\Box$ and $\alpha^\Diamond$, depending on whether it occurs within a box or a diamond modality. This distinction reflects the different roles of the relations $R_\Box(\alpha)$ and $R_\Diamond(\alpha)$ in the constructive semantics, e.g. $a^\Box := i^*;a$ and $a^\Diamond := a$ for $a \in \Pi_0$.
Consequently, the constructive modality $\langle \alpha \rangle$ becomes $[i^*]\langle \alpha^\Diamond\rangle$, which corresponds precisely to the constructive semantics if we interpet $\peq$ as $\rho(i^*)$.
The box modality $[\alpha]$, in turn, translates as $[\alpha^\Box]$, according to the alternative interpretation of the box established in Lemma \ref{AlternativeBox}.
We thus arrive at the following translation.

\begin{definition}
We mutually define the translation of a formula $\varphi \in \mathcal{L}$, denoted $\tau(\varphi)$, and the translations $\alpha^\Box$ and $\alpha^\Diamond$ of a program $\alpha \in \Pi_\mathcal{L}$ as
\begin{enumerate}[label=\textbullet, wide, labelwidth=!, labelindent=0pt, itemsep=0pt]

\item $\tau(\bot) := \bot$,

\item $\tau(p) := [i^*]p$,

\item $\tau(\varphi\odot\psi) := \tau(\varphi)\odot \tau(\psi)$ for $\odot\in \{\wedge,\vee\}$,

\item $\tau(\varphi\to\psi) := [i^*] \big(\tau(\varphi)\to \tau(\psi)\big)$,

\item $\tau([\alpha] \varphi) := [ \alpha ^\Box ] \tau(\varphi)$,

\item $\tau(\langle\alpha\rangle  \varphi) := [i^*] \langle \alpha^\Diamond \rangle \tau(\varphi)$,

\item $a^\Box := i^* ; a$,

\item $a^\Diamond := a$,

\item $(\alpha \odot \beta)^\circ := \alpha^\circ \odot \beta^\circ$ for $\odot \in \{; , \cup\}$ and $\circ \in \{\Box, \Diamond\}$,

\item $(\alpha^*)^\circ := (\alpha^\circ)^*$ for $\circ \in \{\Box, \Diamond\}$,

\item $(\varphi?)^\Box := i^* ; \tau(\varphi) ?$,

\item $(\varphi?)^\Diamond := \tau(\varphi) ?$.
\end{enumerate}
\end{definition}

The translation is only linear on the size of the formula.

\begin{proposition}
\label{PolynomialTranslation}
$|\tau(\varphi)|$ is linear on $|\varphi|$ for any $\varphi \in \mathcal{L}$.
\end{proposition}

\begin{proofidea}
It suffices to see that $|\alpha^\circ| \leq 5 \cdot |\alpha|$ for $\circ \in \{\Box, \Diamond\}$ and $\alpha \in \Pi_\mathcal{L}$, and $|\tau(\varphi)| \leq 5 \cdot |\varphi|$ for $\varphi \in \mathcal{L}$ by an easy by mutual induction. \qedhere
\end{proofidea}

We now show that the translation is exact by transforming a $\sf WPDL$-model into a $\sf PDL$ model preserving truth and falsity.
To be precise, given a $\sf WPDL$-model $\Model = \tuple{W,\peq,\rho,\val\cdot}$, we define $\mathcal{M}^{\sf PDL} := \tuple{W, \rho^{\sf PDL}, \val\cdot}$ to be the $\sf PDL$ model such that $\rho^{\sf PDL}(i) := {\peq}$ and $\rho^{\sf PDL}(a) := \rho(a)$ for any other atomic program $a \neq i$. Exactness is established by mutual induction using the defined model.

\begin{proposition}
\label{propExactnessWKtoDPL}
Let $\Model$ be a $\sf WPDL$-model. Then, for any formula $\varphi\in\mathcal L$ and program and $\alpha \in \Pi_\mathcal{L}$,
\begin{enumerate}
    \item\label{ExactnessToPDLRel} $R_\circ(\alpha) =\rho^{\sf PDL}(\alpha^\circ)$ for $\circ \in \{\Box, \Diamond\}$,

    \item\label{ExactnessToPDL} $\Model,w \Vdash\varphi$ iff $\Model^{\sf PDL},w \Vdash \tau (\varphi)$ for any $w \in W$.
\end{enumerate}
\end{proposition}

\begin{proof}
We prove the statements by mutual induction.

\begin{enumerate}[label=\textbullet, wide, labelwidth=!, labelindent=0pt, itemsep=0pt]
    \item[\ref{ExactnessToPDLRel}.] The cases of the programs are straightforwardly satisfied by definition using the inductive hypothesis. We provide details for test program, leaving the rest of the cases to be easily verified by the reader.
    \begin{enumerate}[label=\textbullet, wide, labelwidth=!, labelindent=10pt, itemsep=0pt]

    \item[($\psi?$)]: First, we show the case over the box relation. By definition, the statement \ref{ExactnessToPDL} and the fact that ${\peq} = {\peq^*}$, we conclude
    \begin{equation*}
        \begin{split}
            w R_\Box(\psi?)v & \ \Longleftrightarrow \ w \peq v \ \wedge \ \mathcal{M}, v \Vdash \psi \ \Longleftrightarrow w \rho^{\sf PDL}(i^*) v \ \wedge \  \mathcal{M}^{\sf PDL}, v \Vdash \tau(\psi) \\ & \ \Longleftrightarrow w \rho^{\sf PDL}(i^*) ; \rho^{\sf PDL}(\tau(\psi)?) v \ \Longleftrightarrow \ w \rho^{\sf PDL}(i^* ; \tau(\psi)?) v \\ & \ \Longleftrightarrow \ w \rho^{\sf PDL}((\psi?)^\Box) v.
        \end{split}
    \end{equation*}
    \end{enumerate}

    For the interpretation over the diamond relation, we simply conclude by definition and statement \ref{ExactnessToPDL}:
    \begin{equation*}
      \begin{split}
          w R_\Diamond(\psi?) v & \ \Longleftrightarrow \ w = v \ \wedge \ \mathcal{M}, w \Vdash \psi \ \Longleftrightarrow \ w = v \ \wedge \ \mathcal{M}^{\sf PDL}, w \Vdash \tau(\psi) \\ & \  \Longleftrightarrow \ w \rho^{\sf PDL}(\tau(\psi)?) v \ \Longleftrightarrow \ w \rho^{\sf PDL}((\psi?)^\Diamond) v.
      \end{split}  
    \end{equation*}

    \item[\ref{ExactnessToPDL}.] We provide details on the atomic and the box formula, leaving the remaining cases to be easily verified by the reader.
    \begin{enumerate}[label=\textbullet, wide, labelwidth=!, labelindent=5pt, itemsep=0pt]
    \item[($p$)]: The right to left implication is trivially satisfied since, assuming $\mathcal{M}^{\sf PDL}, w \Vdash [i^*]p$, we in particular get $w \in \|p\|$. For the inverse implication, assume $\mathcal{M}, w \Vdash p$ and consider $w \rho^{\sf PDL}(i^*) v$, that is, $w \peq v$. Thus, by atomic persistence we show $v \in \|p\|$, so we can conclude $\mathcal{M}^{\sf PDL}, w \Vdash [i^*]p$.

    \item[($\lbrack \alpha \rbrack \varphi$)]: Using the alternative interpretation of $[\alpha]$-formulas in Lemma \ref{AlternativeBox}, by the statement \ref{ExactnessToPDLRel} and the inductive hypothesis we conclude
    \begin{equation*}
    \begin{split}
        \mathcal{M}, w \Vdash [\alpha] \varphi & \ \Longleftrightarrow \ \forall v. \ w R_\Box(\alpha) v \to \mathcal{M}, v \Vdash \varphi \\ & \ \Longleftrightarrow \ \forall v. \ w \rho^{\sf PDL}(\alpha^\Box) v \to \mathcal{M}^{\sf PDL}, v \Vdash \tau(\varphi) \\
        & \ \Longleftrightarrow \ \mathcal{M}^{\sf PDL}, w \Vdash [\alpha^\Box] \tau(\varphi) = \tau([\alpha] \varphi).
    \end{split}
    \end{equation*}
    \end{enumerate}
\end{enumerate}

\end{proof}

Now we must transform a $\sf PDL$ model into a $\sf WPDL$ one to show that our translation is correct.
Let $\mathfrak{M} = \tuple{W,\rho, \val\cdot}$ be a $\sf PDL$-model and define the $\sf WPDL$-model $\mathfrak{M}^{\sf W} :=\tuple{W,\rho(i^*),\rho,\val\cdot^{\sf W}}$ where the valuation is defined as $\|p\|^{\sf W} := \|[i^*] p\|$ for any $p \in \mathrm{Prop}$ and $\|\bot\|^{\sf W} := \emptyset$.

\begin{proposition}\label{propCorrectnessWKtoPDL}
Let $\mathfrak{M}$ be a $\sf PDL$-model. Then, $\mathfrak{M}^{\sf W}$ is a $\sf WPDL$-model and for any formula $\varphi \in \mathcal{L}$ and program $\alpha \in \Pi_\mathcal{L}$,
\begin{enumerate}
    \item\label{CorrectnessToPDLRel} $\rho(\alpha^\circ) = R_\circ(\alpha)$ for $\circ \in \{\Box, \Diamond\}$,
    
    \item\label{CorrectnessToPDL} $\mathfrak{M}^{\sf W},w \Vdash  \varphi $ iff $\mathfrak{M},w \Vdash \tau(\varphi)$ for any $w \in W$.
\end{enumerate}
\end{proposition}

\begin{proof}
To show that $\mathfrak{M}^{\sf W}$ is a $\sf WPDL$-model it suffices to check that the defining conditions hold, which is straightforward and omitted here. As an illustration, we spell out the argument for atomic persistence: assuming $w \rho(i^*) v$ and $w \in \|p\|^{\sf W} = \|[i^*] p\|$ for some $p \in \mathrm{Prop}$, by definition and reflexivity we obtain $v \in \|p\| \subseteq \|[i^*] p\|$. Therefore, we conclude that $v \in \|p\|^{\sf WK}$ as required.

We proceed to show the two statements by mutual induction.

\begin{enumerate}[label=\textbullet, wide, labelwidth=!, labelindent=0pt, itemsep=0pt]
    \item[\ref{CorrectnessToPDLRel}.] The program cases are easily satisfied by definition and the inductive hypotheses. We only provide details on the test program interpreted over the box relation, and leave the rest to be easily verified by the reader.
    \begin{enumerate}[label=\textbullet, wide, labelwidth=!, labelindent=3pt, itemsep=0pt]

    \item[($\psi?$)]: By definition and the statement \ref{CorrectnessToPDL} we conclude
    \begin{equation*}
    \begin{split}
        w \rho((\psi?)^\Box) v & \ \Longleftrightarrow \ w \rho(i^* ; \tau(\psi)?) v \ \Longleftrightarrow \ w \rho(i^*) u \ \wedge \ u \rho(\tau(\psi)?) v \\ & \ \Longleftrightarrow \ w \peq u = v \ \wedge \ \mathfrak{M}, v \Vdash \tau(\psi) \\ & \ \Longleftrightarrow \ w \peq v \ \wedge \ \mathfrak{M}^{\sf W}, v \Vdash \psi \ \Longleftrightarrow \ w R_\Box(\psi?) v.
    \end{split}
    \end{equation*}
    \end{enumerate}

    \item[\ref{CorrectnessToPDL}.] Likewise, the formula cases follow straightforwardly by definition and the inductive hypotheses. Observe that, for the box formula, the result is trivially satisfied by the inductive hypothesis using the alternative interpretation of Lemma \ref{AlternativeBox}.
    Here, we only provide details on the diamond formula for an illustration. Then, by the statement \ref{CorrectnessToPDLRel} and the inductive hypothesis we conclude
    \begin{equation*}
    \begin{split}
        \mathfrak{M}, w \Vdash \tau(\langle \alpha\rangle \varphi) & \ \Longleftrightarrow \ \mathfrak{M}, w \vdash [i^*]\langle \alpha^\Diamond \rangle \tau(\varphi) \\ & \  \Longleftrightarrow \ \forall v. \ w \rho(i^*) v \to \exists u. \ v \rho(\alpha^\Diamond) u \ \wedge \ \mathfrak{M}, u \Vdash \tau(\varphi) \\ & \ \Longleftrightarrow \ \forall v. \ w \peq v \to \exists u. \ v \ R_\Diamond(\alpha) u \ \wedge \ \mathfrak{M}^{\sf W}, u \Vdash \varphi \\ & \ \Longleftrightarrow \ \mathfrak{M}^{\sf W},w \Vdash \langle \alpha \rangle \varphi.  
    \end{split}
    \end{equation*}
\end{enumerate}
\end{proof}

Putting together the correctness and exactness of $\tau$, we conclude that $\tau$ preserves validity.

\begin{theorem}
\label{WKiffPDL}
$\Vdash_{\sf WPDL} \varphi$ iff $\Vdash_{\sf PDL} \tau (\varphi)$ for any $\varphi\in\mathcal L$.
\end{theorem}

\begin{proof}
The left to right implication holds by Proposition \ref{propCorrectnessWKtoPDL}, while the inverse follows by Proposition \ref{propExactnessWKtoDPL}.
\end{proof}

\subsection{Upper bounds for the constructive $\sf PDL$}\label{sec:UpperBoundsMML}

The above results show that $\tau$ is a `proper G\"odel-Tarski translation', and thus $\sf WPDL$ may be seen as a fragment of $\sf PDL$.
The main advantage of having such a translation is that it allows us to import many results from the classical target logic into the intuitionistic source logic. In particular, we inherit the following properties of $\sf WPDL$ directly from $\sf PDL$, and extend them to $\sf CPDL$ and its fragments using the translation of Section \ref{secCW}.

\begin{theorem}
\label{theoEXPTIME}
Validity of ${\sf WPDL}$, $\sf CPDL$, ${\sf CPDL}_\Box$ and $\sf CK^*_{\Box}$ is in {\sc ExpTime}.
\end{theorem}

\begin{proof}
For ${\sf WPDL}$, the claim follows from Theorem~\ref{WKiffPDL}.  
For ${\sf CPDL}$, the same {\sc ExpTime} bound follows via composition with the translation defined in Section \ref{secCW} by Theorem \ref{CorrespondenceConstructiveAndWijesekera}, and immediately also for the fragments ${\sf CPDL}_\Box$ and $\sf CK^*_{\Box}$.
\end{proof}

Given that the transformation $\tau$ is linear on the size of the formula (Proposition \ref{PolynomialTranslation}), $\sf WPDL$ directly inherits the following result from $\sf PDL$. The same extends to ${\sf CPDL}$ via the embedding defined in Section \ref{secCW}, that is shown polynomial on the size of the formula (Proposition \ref{PolynomialOmega}).

\begin{theorem}
\label{FMP}
${\sf WPDL}$, ${\sf CPDL}$, ${\sf CPDL}_\Box$ and $\sf CK^*_{\Box}$ have the exponential-size model property.
\end{theorem}

\begin{proof}
For ${\sf WPDL}$, if $\varphi \in \mathcal{L}$ is satisfiable, then $\tau(\varphi)$ is satisfiable in ${\sf PDL}$ (Theorem~\ref{WKiffPDL}). By the exponential-size model property of ${\sf PDL}$, $\tau(\varphi)$ has a model of size at most exponential in $|\tau(\varphi)|$, which yields a ${\sf WPDL}$-model of the same size satisfying $\varphi$ (Proposition~\ref{propCorrectnessWKtoPDL}).  
For ${\sf CPDL}$, the result follows via composition with the translation defined in Section \ref{secCW} using Theorem~\ref{CorrespondenceConstructiveAndWijesekera}, and the same bound holds immediately for ${\sf CPDL}_\Box$ and $\sf CK^*_{\Box}$.
\end{proof}

\section{Translating $\sf K^*$ into ${\sf CK^*_{\Box}}$}\label{sec:KtoCKbox}

In order to conclude the {\sc ExpTime}-completeness of the constructive ${\sf PDL}$ via the chain of embeddings established in the previous sections, we now provide a translation of $\sf K^*$, the basic $a, a^*$-fragment of $\sf PDL$, into the diamond-free $a, a^*$-fragment of constructive $\sf PDL$.

The key step is to define a translation 
\(\iota : \mathcal{L}^* \to \mathcal{L}^{\Box,*}\) which ensures that if $\varphi$ is valid in a classical $\sf PDL$-model, then its translation is valid in the corresponding $\sf CPDL$-model. To achieve this, we use a `converse G\"odel-Tarski'-style translation that internalizes classical reasoning: \(\iota(\varphi)\) holds in a constructive model precisely when, assuming the excluded middle for all subformulas of \(\varphi\) at any accessible world, then \(\varphi\) itself holds. Observe that only need to consider accessibility over the only atomic program $a$ of our target language, ensuring polynomial size of the translation.

\begin{definition}
The translation $\iota : \mathcal{L}^* \to \mathcal{L}^{\Box,*}$ is defined for any $\varphi \in \mathcal{L}^*$ as 
\[
\iota(\varphi) := [a^*] \bigwedge \{\psi \vee \neg \psi \mid \psi \in Subf(\bar{\varphi})\} \to \bar{\varphi}\] 
where $\bar{\varphi}$ is obtained from $\varphi$ by unfolding every occurrence of a diamond modality $\langle\alpha\rangle$ as $\neg[\alpha]\neg$.
\end{definition}

Note that the unfolding of the diamonds ensure that the translated formula lies within the intended diamond-free fragment $\mathcal{L}^{\Box,*}$. 

As for previous translations, it is key that our embedding is polynomial on the size of the formula.

\begin{proposition}
\label{PolynomialTranslationiota}
$|\iota(\varphi)|$ is polynomial on $|\varphi|$ for any $\varphi \in \mathcal{L}^*$.
\end{proposition}

\begin{proof}
We can easily deduce that $|\{\psi \vee \neg \psi \mid \psi \in Subf(\bar{\varphi})\}| \leq |\bar{\varphi}| \cdot (2 \cdot |\bar{\varphi}| + 2)$ and $|\bar{\varphi}|$ is clearly linear on $|\varphi|$, therefore $|\iota(\varphi)|$ is quadratic on $|\varphi|$.
\end{proof}

To establish the exactness of the translation, we rely on a stronger observation: a formula $\phi \in \mathcal{L}^{\Box, *}$ is satisfied in a classical ${\sf PDL}$ model if and only if it is satisfied in the corresponding $\sf CPDL$-model obtained by preserving the program interpretation and taking the intuitionistic relation to be trivial. As a result, the constructive model behaves classically. 

\begin{proposition}
\label{exactnessKCK}
Let $\mathfrak{M} = \tuple{W, \rho, \|\cdot\|}$ be a $\sf PDL$-model. Define $\mathfrak{M}^{\sf C} := \tuple{W, \mathrm{Id}, \rho, \|\cdot \|^{\sf C}}$ where $\|p\|^{\sf C} := \|p\|$ for any $p \in \mathrm{Prop}$, and $\|\bot\|^{\sf C} := \emptyset$.
Then, $\mathfrak{M}^{\sf C}$ is a $\sf CPDL$-model and $\mathfrak{M}, w \Vdash \phi$ iff $\mathfrak{M}^{\sf C}, w \Vdash \phi$ for any formula $\phi \in \mathcal{L}^{\Box,*}$ and $w \in W$.
\end{proposition}

\begin{proof}
By a straightforward verification of the model conditions, $\mathfrak{M}^{\sf C}$ is a ${\sf CPDL}$-model; the details are omitted. For the second part of the proof we proceed by induction on $\phi \in \mathcal{L}^{\Box, *}$.
The formula cases hold easily by the inductive hypotheses and the definitions. In particular, the case of the implication formula is trivially satisfied given the trivial intuitionistic relation. We provide details for the box formula and leave the rest to be easily verified by the reader. Observe that $\rho(\alpha) = R_\Box^{\mathfrak{M}^{\sf C}}(\alpha)$ is trivially satisfied for any $\alpha \in \Pi_{\mathcal{L}^{\Box,*}}$. Therefore, applying the inductive hypothesis and using the alternative interpretation of the box given in Lemma \ref{AlternativeBox}, we conclude that
    \begin{equation*}
    \begin{split}
        \mathfrak{M}, w \Vdash [\alpha] \phi & \ \Longleftrightarrow \ \forall v. \ w \rho(\alpha) v \to \mathfrak{M}, v \Vdash \phi \\ & \ \Longleftrightarrow \ \forall v. \ w R_\Box^{\mathfrak{M}^{\sf C}}(\alpha) v \to \mathfrak{M}^{\sf C}, v \Vdash \phi \ \Longleftrightarrow \ \mathfrak{M}^{\sf C}, w \Vdash [\alpha] \phi.
    \end{split}
    \end{equation*}
\end{proof}

For correctness, we focus on preserving satisfiability only for the subformulas of the translated formula, and only within the constructive submodel generated by the worlds where excluded middle holds for all these subformulas. Since validity in the diamond-free fragment is independent of model fallibility, we assume without loss of generality that the source model is a $\sf WPDL$-model to avoid dealing with fallible worlds.

Let $\mathcal{M} =\tuple{W, \peq, \rho, \|\cdot\|}$ be a $\sf WPDL$-model and $\varphi \in \mathcal{L}^*$ a formula. Given $\bar{\varphi}$ obtained by unfolding the diamond notation of $\varphi$, consider 
\[
U_{\varphi}^\mathcal{M} := \big\|[a^*] \bigwedge \{\psi \vee \neg \psi \mid \psi \in Subf(\bar{\varphi})\}\big\|.
\] 
We write $\mathcal{M}_{\varphi}$ to denote the submodel of $\mathcal{M}$ generated by $U^\mathcal{M}_{\varphi}$, i.e., the bi-relational model obtained by restricting $\mathcal{M}$ to the closure of $U^\mathcal{M}_{\varphi}$ under $\peq$ and $\rho$. Over the worlds in the generated submodel $\mathcal{M}_\varphi$, we can show that satisfiability of $\varphi$ coincides with that of $\mathcal{M}$.

\begin{lemma}
\label{SatGeneratedModel}
Let $\mathcal{M}$ be a ${\sf WPDL}$-model and $\varphi \in \mathcal{L}^{*}$. Then, $\mathcal{M}, w \Vdash \bar{\varphi}$ iff $\mathcal{M}_\varphi, w \Vdash \bar{\varphi}$ for any $w \in \mathcal{M}_\varphi$.
\end{lemma}

\begin{proof}
It suffices to show a more general statement for the subformulas: $\mathcal{M}, w \Vdash \psi$ iff $\mathcal{M}_{\varphi}, w \Vdash \psi$ for any $\psi \in {\rm Subf}(\bar{\varphi})$.
We proceed by an easy induction on $\psi \in \mathcal{L}^{\Box,*}$, whose details are left to the reader. As an illustration, let us provide details for the box formula. To show this case, first observe that $w R_\Box^\mathcal{M}(\alpha) v$ iff $w R_\Box^{\mathcal{M}_{\varphi}}(\alpha) v$ for any $w \in \mathcal{M}_{\varphi}$ by an easy induction on $\alpha \in \Pi_{\mathcal{L}^{\Box,*}}$. Therefore, by the inductive hypothesis, and using the alternative interpretation of box given by Lemma \ref{AlternativeBox}, we conclude that
\begin{equation*}
    \begin{split}
        \mathcal{M}, w \Vdash [\alpha] \psi & \ \Longleftrightarrow \ \forall v. \ w R_\Box^\mathcal{M}(\alpha) v \to \mathcal{M}, v \Vdash \psi \\ & \ \Longleftrightarrow \ \forall v. \ w R_\Box^{\mathcal{M}_{\varphi}}(\alpha) v \to \mathcal{M}_{\varphi} \Vdash \psi \ \Longleftrightarrow \ \mathcal{M}_{\varphi}, w \Vdash [\alpha] \psi.
    \end{split}
\end{equation*}
\end{proof}

For proving the correctness of the translation, we consider the classical model $\mathcal{M}^{\sf K^*}_{\varphi}$ obtained by restricting $\langle W, (\peq ; \rho), \|\cdot\| \rangle$ to the set of worlds of $\mathcal{M}_{\varphi}$, omitting $\bot$ from the valuation domain.

\begin{proposition}
\label{correctenessKCK}
Let $\mathcal{M}$ be a $\sf WPDL$-model, $\varphi \in \mathcal{L}^*$, and $w \in \mathcal{M}_\varphi$. Then, $\mathcal{M}^{\sf K^*}_\varphi, w \Vdash \varphi$ iff $\mathcal{M}_\varphi, w \Vdash \bar{\varphi}$.
\end{proposition}

\begin{proof}
Given that satisfiability of $\varphi$ is equivalent to that of $\bar{\varphi}$ over classical models, we show a more general statement over the subformulas: $\mathcal{M}_{\varphi}^{\sf K^*}, w \Vdash \psi$ iff $ \mathcal{M}_{\varphi}, w \Vdash \psi$
for any $\psi \in {\rm Subf}(\bar{\varphi})$ and $w \in \mathcal{M}_{\varphi}$.
We proceed by induction on $\psi \in \mathcal{L}^{\Box, *}$, providing details on the cases of negation and box formulas, and leaving the rest to be easily verified by the reader.
\begin{enumerate}[label=\textbullet, wide, labelwidth=!, labelindent=5pt, itemsep=0pt]
    \item[($\neg \psi$)]: First, assume that $\mathcal{M}_{\varphi}^{\sf K^*}, w \Vdash \neg \psi$, that is,  $\mathcal{M}_{\varphi}^{\sf K^*}, w \not\Vdash \psi$. By the inductive hypothesis we have that  $\mathcal{M}_{\varphi}, w \not\Vdash \psi$. Therefore, since $w \in \mathcal{M}_{\varphi}$ is defined to satisfy the excluded middle for the subformulas of $\bar{\varphi}$, we conclude that  $\mathcal{M}_{\varphi}, w \Vdash \neg \psi$.

    Conversely, now assume that $\mathcal{M}_{\varphi}, w \Vdash \neg \psi$. It suffices to show that $\mathcal{M}_{\varphi}, w \not\Vdash \psi$ so we directly conclude by the inductive hypothesis that $\mathcal{M}^{\sf K^*}_{\varphi} \Vdash \neg \psi$. This is direct since, if $\mathcal{M}_{\varphi}, w \Vdash \psi$ was the case, then $w \in \|\bot\|$ would follow from the assumption, contradicting infallibility of the model.

    \item[($\lbrack \alpha \rbrack \psi$)]: To show this formula case it suffices to see that ${(\peq ; \rho)(\alpha)} = {R_\Box(\alpha)}$ for any program $\alpha \in \Pi_{\mathcal{L}^{\Box,*}}$ by an easy induction on $\alpha$, whose details are here omitted. Therefore, by the inductive hypothesis, and using the alternative interpretation of box given in Lemma \ref{AlternativeBox}, we conclude that
    \begin{equation*}
        \begin{split}
            \mathcal{M}_{\varphi}, w \Vdash [\alpha] \psi & \ \Longleftrightarrow \ \forall v. \ w R_\Box(\alpha) v \to \mathcal{M}_{\varphi}, v \Vdash \psi \\ & \ \Longleftrightarrow \ \forall v. \ w (\peq;\rho)(\alpha) v \to \mathcal{M}_{\varphi}^{\sf K^*}, v \Vdash \psi \ \Longleftrightarrow \ \mathcal{M}_{\varphi}^{\sf K^*}, w \vdash [\alpha] \psi.
        \end{split}
    \end{equation*}
\end{enumerate}

\end{proof}

By the preceding results we conclude that the translation preserves validity as intended.

\begin{theorem}
\label{KiffCK}
$\Vdash_{\sf K^*} \varphi$ iff $\Vdash_{\sf CK^*_\Box} \iota (\varphi)$ for any $\varphi\in\mathcal L^{*}$.
\end{theorem}

\begin{proof}
For exactness, suppose $\not\Vdash_{\sf K^*} \varphi$, that is, there is a ${\sf K^*}$-model $\mathfrak{M}$ and $w \in \mathfrak{M}$ such that $\mathfrak{M}, w \not\Vdash \varphi$. Given that $\mathfrak{M}$ is classical, it is equivalent to $\mathfrak{M}, w \not\Vdash \bar{\varphi}$. Moreover, we also have that $\mathfrak{M}, w \Vdash [a^*] \bigwedge \{ \psi \vee \neg \psi \mid \psi \in Subf(\bar{\varphi}) \}$ is trivially satisfied. Thus, we deduce $\mathfrak{M}, w \not\Vdash \iota(\varphi)$ by definition, and conclude $\mathfrak{M}^{\sf C}, w \not\Vdash \iota(\varphi)$ by Proposition~\ref{exactnessKCK}.

Conversely, assume $\mathcal{M}, w \not\Vdash \iota(\varphi)$ for a ${\sf CPDL}$-model $\mathcal{M}$ over $\mathcal{L}^{\Box, *}$ and a world $w \in \mathcal{M}$. By Lemma \ref{EqCKWKBox}, there is a $\sf WPDL$-model $\mathcal{N}$ such that $\mathcal{N}, v \not\Vdash \iota(\varphi)$ at some $v \in \mathcal{N}$, i.e. there exists $u \seq^\mathcal{N} v$ such that
$\mathcal{N}, u \Vdash  [a^*] \bigwedge \{ \psi \vee \neg \psi \mid \psi \in Subf(\bar{\varphi}) \}$ and $\mathcal{N}, u \not\Vdash \bar{\varphi}$. In particular, that means that $u \in \mathcal{N}_\varphi$ so we get $\mathcal{N}_\varphi, u \not\Vdash \bar{\varphi}$ by Lemma \ref{SatGeneratedModel}. Thus, by Proposition \ref{correctenessKCK} we conclude that $\mathcal{N}_\varphi^{\sf K^*}, u \not\Vdash \varphi$. \qedhere
\end{proof}

We thus obtain lower bounds for the constructive ${\sf PDL}$.

\begin{theorem}
\label{theoEXPTIMEhard}
Validity of ${\sf WPDL}$, $\sf CPDL$, ${\sf CPDL}_\Box$ and ${\sf CK}^*_\Box$ is {\sc ExpTime}-hard.
\end{theorem}

\begin{proof}
For ${\sf CK}^*_\Box$, the claim follows directly from Theorem~\ref{KiffCK}. By simple embeddings of the fragments, the same lower bound extends to $\sf CPDL_\Box$ and then to $\sf CPDL$. For ${\sf WPDL}$, the {\sc ExpTime} bound carries over via the translation defined in Section \ref{secCW} using Theorem~\ref{CorrespondenceConstructiveAndWijesekera}.
\end{proof}

Combined with the upper bounds of Theorem~\ref{theoEXPTIME}, this yields {\sc ExpTime}-completeness of the constructive $\sf PDL$.

\begin{corollary}
${\sf WPDL}$, $\sf CPDL$, ${\sf CPDL}_\Box$ and $\sf CK^*_\Box$ are {\sc ExpTime}-complete.
\end{corollary}

\section{An application: $\sf CS4$ and $\sf WS4$}\label{CWS4}

If we restrict the programs of the language, it is easy to see that the presented logics $\sf CPDL$ and $\sf WPDL$ extend $\sf CK$ and $\sf WK$ respectively, and thus our results also yield {\sc ExpTime} upper bounds on their decision problems.
However, this observation is not too interesting since a simpler translation into ${\sf S4}\oplus {\sf K}$ is available, yielding a sharper {\sc PSpace} bound~\cite{Dalmonte25}.
In contrast, if we focus only on $a^*$-programs, we may also view our logics as extensions of constructive $\sf S4$, which do not seem to embed into ${\sf S4}\oplus {\sf S4}$.
We thus turn our attention to transitive, reflexive logics over the \emph{intuitionistic modal language} $\mathcal L^{\rm IML}$, built up from $\bot$ and propositional variables using $\wedge$, $\vee$, $\to$, $\Box$ and $\Diamond$. Observe that we can equivalently think of $\mathcal{L}^{\rm IML}$ as the fragment of $\mathcal{L}$ whose only program is $a \in \Pi_0$, so we can write $\langle a \rangle$ and $[a]$ as $\Diamond$ and $\Box$, respectively.

Let us review the bi-relational semantics of constructive $\sf S4$~\cite{AlechinaMPR01}.

\begin{definition}\label{def::csf-model}
A $\sf CS4$-frame, or bi-preorder, is a tuple $\mathcal F=\tuple{\domai \Frame ,\peq^\Frame, \R^\Frame}$ where 
    \begin{itemize}
        \item $\domai{\Frame}$ is the set of {\em possible worlds},
        
        \item the {\em intuitionistic relation} $\peq^{\Frame}$ is a preorder,

        \item the {\em modal relation} $R^\mathcal{F} \subseteq 2^{W^\mathcal{F} \times W^\mathcal{F}}$ is a preorder,

        \item $\Frame$ is {\em confluent,} i.e.~$w\R^\Frame v \peq^\Frame v'$ implies that there is $w'$ such that $w\peq^\Frame w' \R^\Frame v'$.
\end{itemize}

A {\em ${\sf CS4}$-model} is a tuple $\Model=\tuple{\domai{\Model}, \peq^\Model, R^\Model, \val\cdot^{\Model}}$ consisting of a $\sf CS4$-frame equipped with a $\sf CS4$-valuation $\val\cdot^{\Model}:\mathrm{Prop} \cup \{\bot\}\to\mathcal P(\domai{\Model})$ such that, for any $p\in\mathrm{Prop}$,
\begin{enumerate}
 
 \item
 $\val\bot^\Model \subseteq\val p^\Model$ \quad 
 \textit{(atomic ex falso quodlibet),}
 
 \item if $w\peq^\mathcal{M} v$ and $w\in \val p^{\Model}$ then $v\in \val p^{\Model}$ \quad \textit{(atomic persistence),}

  \item if $w\in \val\bot^\Model $ and $w\peq^\mathcal{M} v$ or $w R^\mathcal{M} v$, then $v\in \val\bot^\Model$ \quad \textit{(falsum persistence),}

 \item\label{itCondSerialBox} if $w\in \val\bot^\Model $ then there exists $v$ such that $w R^\mathcal{M} v$ \quad \textit{(falsum seriality).}

\end{enumerate}

A \emph{${\sf WS4}$-model} is an infallible ${\sf CS4}$-model.
\end{definition}

The confluence property can be found in the literature under the name \textit{backwards confluence} or \textit{back-up confluence} (see e.g.~\cite{BalbianiToCL,balbiani2024variants}). The satisfiability relation $\Vdash$ and validity notions apply naturally to $\mathcal{L}^{\rm IML}$ over the defined structures, with the following standard definitions for the modalities:
\begin{itemize}
    \item $\mathcal{M}, w \Vdash \Box \varphi \ \Longleftrightarrow \ \forall v \seq w. \ \forall u. \ v R u \to \mathcal{M}, u \Vdash \varphi$,

    \item $\mathcal{M}, w \Vdash \Diamond \varphi \ \Longleftrightarrow \ \forall v \seq w. \ \exists u. \ v R u \ \wedge \ \mathcal{M}, u \Vdash \varphi$,   
\end{itemize}

We write $\Vdash_{\sf CS4} \varphi$ (respectively, $\Vdash_{\sf WS4} \varphi$) to mean that $\varphi \in \mathcal{L}^{\rm IML}$ is valid on all $\sf CS4$-frames ($\sf WK4$-frames, respectively).

\subsection{Embedding transitive extensions into the constructive $\sf PDL$}\label{subsec:embedding-wk} 

One way to view the $\sf CS4$ semantics for $\Box$ is as a standard modality, albeit based on a new relation, $(\peq;R)$.
The reason that confluence
enforces the $\ax{4}_\Box$ axiom is that it makes $(\peq;\R)$ be a preorder, provided $\R$ is.
This ensures the validity of $\Box p\to \Box\Box p$ over $\sf CS4$ models, which does not follow from the transitivity of $\R$ alone.
This is well known~\cite{AlechinaMPR01}, but we include a proof to keep the article self-contained.

\begin{center}
\begin{minipage}{0.60\textwidth}
\raggedright
\begin{lemma}
\label{leqRCS4}
If $\Model$ is a $\sf CS4$-frame, then $(\peq;R)$ is a preorder.
\end{lemma}

\begin{proof}
We have that $w \peq w \R w$ by reflexivity of both relations. For transitivity, assume $w (\peq;R) v$ and $v (\peq;R) u$ for some $w, v, u \in W^\mathcal{M}$. By definition, we can find $w'$ and $v'$ such that $w \peq w' R v$ and $v \peq v' R u$. In particular, from $w' R v \peq v'$, by confluence there exists $w''$ satisfying $w' \peq w'' R v'$. Thus, by transitivity of the relations we conclude $w \peq w'' R u$. \qedhere
\end{proof}
\end{minipage}
\hfill
\begin{minipage}{0.35\textwidth}
    \centering
\begin{tikzpicture}
    \begin{pgfonlayer}{nodelayer}
        \node [style=root] (0) at (0, 0) {$w$};
        \node [style=root] (1) at (0, 1.2) {$w'$};
        \node [style=root] (2) at (1.2, 1.2) {$v$};
        \node [style=root] (3) at (1.2, 2.4) {$v'$};
        \node [style=root] (4) at (2.4, 2.4) {$u$};
        \node [style=root] (5) at (0, 2.4) {$w''$};
    \end{pgfonlayer}
    \begin{pgfonlayer}{edgelayer}
        \draw [style=edgeto, red] (0) to (1);
        \draw [style=edgeto, blue] (1) to (2);
        \draw [style=edgeto, red] (2) to (3);
        \draw [style=edgeto, blue] (3) to (4);

        \draw [style=edgeto] (0) to node[midway, sloped, below] {$\peq ; R$}  (2);

        \draw [style=edgeto] (2) to node[midway, sloped, below] {$\peq ; R$}  (4);
        
        \draw [style=edgeto, dashed, red] (1) to (5);
        \draw [style=edgeto, dashed, blue] (5) to (3);

        \draw [style=edgeto, red, bend left=30] (0) to (5);

        \draw [style=edgeto, blue, bend left=30] (5) to (4);
    \end{pgfonlayer}
\end{tikzpicture}
\captionof{figure}{Confluence in Lemma \ref{leqRCS4}.}
\label{fig:LemmaCK4peqreltrans}
\end{minipage}
\end{center}

This suggests a translation $\kappa : \mathcal{L}^{\rm IML} \to \mathcal{L}$ of $\sf CS4$ into ${\sf CPDL}$, where $\kappa(\varphi)$ is the result of replacing every occurrence of $\Box$ by $[a^*]$ and every occurrence of $\Diamond$ by $\langle a^*\rangle$. 
Observe that the translation lies into the $a^*$-fragment of $\mathcal{L}$ and, since it is defined by direct formula replacement, the size of $\kappa(\varphi)$ is clearly polynomial in the size of $\varphi$.

\begin{proposition}
\label{ExactnessCS4toCK}
Let $\mathcal{M} = (W, \peq, R, \|\cdot\|)$ be a $\sf CS4$-model. Define $\mathcal{M}^{\sf CPDL} := (W, \peq, \rho, \|\cdot\|)$ where $\rho(a) := R$ and $\rho(b) = {\rm Id}$ for any atomic program $b \neq a$. Then, $\mathcal{M}^{\sf CPDL}$ is a ${\sf CPDL}$ model and $\mathcal{M},w \Vdash  \varphi$ iff $\mathcal{M} , w \Vdash  \kappa(\varphi)$ for any $\varphi \in \mathcal{L}^{\rm IML}$ and $w \in \mathcal{M}$.
\end{proposition}

\begin{proof}
To see that $\mathcal{M}^{\sf CPDL}$ is a $\sf CPDL$-model, it suffices to check the conditions which are directly inherit from $\mathcal{M}$; the details are left to the reader to verify.

The second part of the proof proceeds by induction on $\varphi \in \mathcal{L}^{\rm IML}$. Here, as an illustration, we provide details on the modality formulas. The rest of the formulas simply follow by definition using the inductive hypotheses.
\begin{enumerate}[label=\textbullet, wide, labelwidth=!, labelindent=5pt, itemsep=0pt]
        \item[($\Box \varphi$)]: First, assume that $\mathcal{M},w \Vdash \Box\varphi$. To show $\mathcal{M}^{\sf CPDL},w \Vdash [a^*]\kappa(\varphi)$, using the alternative interpretation of the box given in Lemma \ref{AlternativeBox} it suffices to see that $\mathcal{M}^{\sf CPDL},v \Vdash \kappa(\varphi)$ for any $v\in W$ such that $w R_\Box(a^*) v$, that is, $w (\peq \cup (\peq ; \rho(a))^* ) v$. If $w \peq v$, from the assumption we get $\mathcal{M}, v \Vdash \varphi$ since the modal relation is a preorder, and conclude $\mathcal{M}^{\sf CPDL}, v \Vdash \kappa(\varphi)$ by the inductive hypothesis. If $w (\peq ; \rho(a))^* v$, we deduce $w(\peq ; R)^* v $ by definition. Given that the composition is also a preorder as shown in Lemma \ref{leqRCS4}, again from the assumption we obtain $\mathcal{M}, v \Vdash \varphi$ and conclude $\mathcal{M}^{\sf CPDL}, v \Vdash \kappa(\varphi)$ by the inductive hypothesis.

    Conversely, suppose that $\mathcal{M}^{\sf CPDL}, w \Vdash [a^*] \kappa (\varphi)$, and consider $v \in W$ such that $w (\peq ; R) v$. Since the composition is a preorder by Lemma \ref{leqRCS4}, we deduce that $w (\peq ; R)^* v$ and, moreover, $w \peq \cup (\peq ; R)^* v$. By definition we have that $w \peq \cup (\peq ; \rho(a))^* v$, that is, $w R_\Box(a^*) v$. Therefore, from the assumption using the alternative interpretation of the box in Lemma \ref{AlternativeBox} we obtain $\mathcal{M}^{\sf CPDL}, v \Vdash \kappa(\varphi)$ and conclude by the inductive hypothesis that $\mathcal{M}, v \Vdash \varphi$.
    
    \item[($\Diamond \varphi$)]: Since $R$ is a preoder, by the inductive hypothesis we conclude that
    \begin{equation*}
    \begin{split}
        \mathcal{M}, w \Vdash \Diamond \varphi & \ \Longleftrightarrow \ \forall v \seq w. \ \exists u. \ v R u \ \wedge \mathcal{M}, u \Vdash \varphi \\ & \ \Longleftrightarrow \ \forall v \seq w. \ \exists u. \ v R^* u \ \wedge \mathcal{M}, u \Vdash \varphi \\ & \ \Longleftrightarrow \ \forall v \seq w. \ \exists u. \ v R_\Diamond(a^*) u \ \wedge \mathcal{M}^{\sf CPDL}, u \Vdash \kappa(\varphi) \\ & \ \Longleftrightarrow \ \mathcal{M}^{\sf CPDL}, w \Vdash \langle a^* \rangle \kappa(\varphi) = \kappa(\Diamond \varphi).
    \end{split}
    \end{equation*}
\end{enumerate}
\end{proof}

To see that the translation is correct, we describe how a $\sf CPDL$-frame $\mathcal{F} = (W, \peq, \rho)$ for which $\rho(a)$ is assumed to be a preorder can be transformed into a $\sf CS4$-frame. 
This assumption on the interpretation of $a$ is without loss of generality: since $\kappa(\varphi)$ contains only the program $a^*$, its validity is unaffected by replacing $\rho(a)$ with its reflexive-transitive closure, as we will note in the final validity preservation theorem.
Therefore, we define $\mathcal{F}^{\sf CS4} := \tuple{W^{\sf CS4},\peq^{\sf CS4},R^{\sf CS4}}$ where
\begin{enumerate}[label=\textbullet, wide, labelwidth=!, labelindent=0pt, itemsep=0pt]

\item $W^{{\sf CS4}} = W\times\{0,1\}$;

\item $(w,i)\peq^{\sf CS4} (v,j) $ if $ w \peq v$;

\item $(w,i) R^{\sf CS4} (v,j)$ if either $i=j=0$ and $w R_\Diamond(a^*) v$, or $i = 1$ and $w R_\Box(a^*) v$.

\end{enumerate}

The $\sf CS4$-frame is defined so that each world is duplicated and the relation $R^{\sf CS4}$ simply uses the interpretation of $a^*$: 
$0$-worlds are related using its diamond interpretation, while $1$-worlds are related following the box interpretation.

The following general lemma for $\sf CPDL$-models guarantees that, independently of the index of the worlds in $\mathcal{M}^{\sf CS4}$, one $R^{\sf CS4}$-step implies one $a^*$-step over the box relation. This simplifies the subsequent case analysis in the correctness proofs. 

\begin{lemma}\label{RDiamondtoRBox}
Let $\mathcal{M}$ be a ${\sf CPDL}$-model. Then, $R_\Diamond(a^*) \subseteq R_\Box(a^*)$ for any $a \in \Pi_0$.
\end{lemma}

\begin{proof}
Given that $\peq$ is a preorder we show that $R_\Diamond(a^*) = \rho(a)^* \subseteq (\peq ; \rho(a))^* = R_\Box(a)^*$.
\end{proof}

Therefore, our construction yields a ${\sf CS4}$-frame for which confluence is guaranteed by the $1$-worlds given that $\rho(a)$ is a preorder on the original ${\sf CPDL}$-frame.

\begin{lemma}\label{lemmaCS4frame}
If $\mathcal{F}$ is a $\sf CPDL$-frame where $\rho^\mathcal{F}(a)$ is a preorder, then $\mathcal{F}^{\sf CS4}$ is a $\sf CS4$-frame.
\end{lemma}

\begin{proof}
The relation $\peq^{\sf CS4}$ is a preorder, inheriting reflexivity and transitivity from $\peq$. Similarly, by definition we can easily show that $R^{\sf CS4}$ is reflexive (by a case distinction on the world index).

To show that $R^{\sf CS4}$ is a preorder, we now provide the proof of transitivity. Hence, assume that $(w, i) R^{\sf CS4} (v, j) R^{\sf CS4} (u, k)$. If the first step is due to $i = j = 0$ and $w R_\Diamond(a^*) v$, we forcely have that $k = 0$ and $v R_\Diamond(a^*) u$ too, so we conclude by definition and transitivity that $(w, i) R^{\sf CS4} (u, k)$. Otherwise, if $i = 1$ and $w R_\Box(a^*) v$, by Lemma \ref{RDiamondtoRBox} we see that any of the cases of the second relation imply that $v R_\Box(a^*) u$, so by transitivity we likewise conclude that $(w,i) R^{\sf CS4} (u, k)$.

For confluence, assuming $(w,i) R^{\sf CS4} (v,j) \peq^{\sf CS4} (u,k)$ we will see that $(w, i) \peq^{\sf CS4} (w,1) R^{\sf CS4} (u, k)$. 
By definition it is clear that $(w, i) \peq^{\sf CS4} (w,1)$. By Lemma \ref{RDiamondtoRBox} we deduce that any of the cases of $(w,i) R^{\sf CS4} (v,j)$ imply $w R_\Box(a^*) v$, so we get $w R_\Box(a^*) v \peq u$ composing with the assumption. Thus, since $\rho(a)$ is a preorder we have that 
\[
R_\Box(a^*) ; \peq \ = \ (\peq ; \rho(a))^* ; \peq \ \subseteq \ (\peq ; \rho(a))^* ; (\peq ; \rho(a)) \ \subseteq \ (\peq ; \rho(a))^* \ = \ R_\Box(a^*).
\]
Thus, we conclude that $(w,1) R^{\sf CS4} (u,k)$ by definition.
\end{proof}

We now enrich the previous construction with a suitable valuation to conclude with the correctness of the translation. For a $\sf CPDL$-model $\mathcal{M} = (W, \peq, \rho, \|\cdot\|)$ where $\rho(a)$ is a preorder, we define $\mathcal{M}^{\sf CS4} := \tuple{W^{\sf CS4}, \peq^{\sf CS4}, R^{\sf CS4}, \|\cdot\|^{\sf CS4}}$ for $\val\cdot^{\sf CS4} :=(\pi^{-1} \circ \val\cdot)$ where $\pi\colon W^{\sf CS4} \to W$ denotes the projection onto the first component.

\begin{proposition}
\label{CorrectnessCS4toCK}
Let $\mathcal{M}$ be a $\sf CPDL$-model where $\rho^\mathcal{M}(a)$ is a preorder. Then,  $\mathcal{M}^{\sf CS4}$ is a $\sf CS4$-model and 
$\mathcal{M},w \Vdash  \kappa(\varphi) $ iff $\mathcal{M}^{\sf CS4} ,(w,i) \Vdash  \varphi$ for any $i \in \{0,1\}$, $w\in \mathcal{M}$ and $\varphi \in \mathcal{L}^{\rm IML}$. 
\end{proposition}

\begin{proof}
By Lemma~\ref{lemmaCS4frame}, the frame conditions are satisfied. It remains to verify the model conditions for $\mathcal{M}^{\sf CS4}$, which are readily inherited from those of $\mathcal{M}$ and left to the reader to verify. As an illustrative case, consider falsum seriality:
assuming $(w, i) \in \|\bot\|^{\sf CS4}$, by definition we simply conclude that $(w,i) R^{\sf CS4} (w, 0)$ for any $i \in \{0,1\}$.

For the second part of the proof, we argue by induction on $\varphi \in \mathcal{L}^{\rm IML}$. We detail only the modality formulas, leaving the remaining ones to be easily verified by the reader.

\begin{enumerate}[label=\textbullet, wide, labelwidth=!, labelindent=0pt, itemsep=0pt]
    \item[($\Box \varphi$)]: First assume that $\mathcal{M}, w \Vdash \kappa(\Box \varphi) = [a^*] \kappa(\varphi)$. In order to show $\mathcal{M}^{\sf CS4}, (w, i) \Vdash \Box \varphi$ for $i \in \{0,1\}$, consider $(w,i) \peq^{\sf CS4} (v,j) R^{\sf CS4} (u, k)$. By definition, that means that $w \peq v$ and, by Lemma \ref{RDiamondtoRBox}, $v R_\Box(a^*) u$. Therefore, we get $\mathcal{M}, u \Vdash \kappa(\varphi)$ by the assumption, and we conclude $\mathcal{M}^{\sf CS4}, (u, k) \Vdash \varphi$ by the inductive hypothesis.

    For the inverse implication, assume $\mathcal{M}^{\sf CS4}, (w, i) \Vdash \kappa(\varphi)$ and consider $w \peq v R_\Box(m^*) u$. By definition that implies that $(w, i) \peq^{\sf CS4} (v, 1) R^{\sf CS4} (u, k)$, so by the assumption we get $\mathcal{M}^{\sf CS4}, (u, k) \Vdash \varphi$ for any $k \in \{0,1\}$. Thus, by the inductive hypothesis we obtain $\mathcal{M}, u \Vdash \kappa(\varphi)$, concluding the case.

    \item[($\Diamond \varphi$)]: First assume that $\mathcal{M}, w \Vdash \kappa(\Diamond \varphi) = \langle a^* \rangle \kappa(\varphi)$. To show that $\mathcal{M}^{\sf CS4}, (w, i) \Vdash \Diamond \varphi$ for $i \in \{0,1\}$, assume that $(v, j) \seq^{\sf CS4} (w, i)$. By definition, that means that $v \seq w$ so, by the inductive hypothesis there is some $u \in \mathcal{M}$ such that $v R_\Diamond(a^*) u$ and $\mathcal{M}, u \Vdash \kappa(\varphi)$. If $j = 0$, we have $(v, 0) R^{\sf CS4} (u, 0)$ by definition and $\mathcal{M}^{\sf CS4}, (u,0) \Vdash \varphi$ by the inductive hypothesis, concluding the case. Otherwise, if $j = 1$ by Lemma \ref{RDiamondtoRBox} we deduce that $v R_\Box(a^*) u$ so we obtain $(v, 1) R^{\sf CS4} (u, k)$ by definition and $\mathcal{M}^{\sf CS4}, (u, k) \Vdash \varphi$ by the inductive hypothesis, concluding the proof. \qedhere
\end{enumerate}
\end{proof}

We conclude that the constructive variants of $\sf S4$ may be seen as the $a^*$-fragment of the constructive variants of $\sf PDL$.

\begin{theorem}
\label{CS4iffCKStar}
$\Vdash_{\sf CS4} \varphi$ iff $\Vdash_{\sf CPDL} \kappa (\varphi)$ and $\Vdash_{\sf WS4} \varphi$ iff $\Vdash_{\sf WPDL} \kappa (\varphi)$ for any $\varphi\in\mathcal L^{\rm IML}$.
\end{theorem}

\begin{proof}
For $\sf CS4$ and $\sf CPDL$, the right to left implication is due to Proposition \ref{ExactnessCS4toCK}. For the inverse implication, assume that $\mathcal{M}, w \not\Vdash \kappa(\varphi)$ for some $\sf CPDL$-model $\mathcal{M}$. By an easy induction on $\varphi \in \mathcal{L}^{\sf IML}$ that is left to the reader, one can show that $\mathcal{M}, w \Vdash \kappa(\varphi)$ iff $\mathcal{N}, w \Vdash \kappa(\varphi)$, where $\mathcal{N}$ is obtained from $\mathcal{M}$ by uniquely updating the valuation of $a$ to be $\rho^\mathcal{N}(a) := \rho^\mathcal{M}(a)^*$. This proof relies on the observation that $R_\circ^\mathcal{M}(a^*) = R_\circ^\mathcal{N}(a^*)$ for any $\circ \in \{\Box, \Diamond\}$, and the fact that the only program occurring in a translated formula is $a^*$. Therefore, we have that $\mathcal{N}, w \not\Vdash \kappa(\varphi)$ and, by Proposition \ref{CorrectnessCS4toCK}, we conclude that $\mathcal{N}^{\sf CS4}, w \not\Vdash \kappa(\varphi)$.

For $\sf WK4$ and $\sf WPDL$, observe that we can extend Propositions \ref{ExactnessCS4toCK} and \ref{CorrectnessCS4toCK} trivially to see that the corresponding defined models are infallible given that the original models are. Therefore, we can trivially extend the argument to $\sf WPDL$ and $\sf WS4$.  \qedhere 
\end{proof}

With this, Theorem~\ref{theoEXPTIME} immediately yields an {\sc ExpTime} upper bound for $\sf CS4$ and its Wijesekera variant $\sf WS4$, improving previously known upper bounds.

\begin{theorem}
    Validity for $\sf CS4$ and $\sf WS4$ is in {\sc ExpTime}.
\end{theorem}

\begin{proof}
By Theorem \ref{CS4iffCKStar}, using the {\sc ExpTime}-validity of $\sf CPDL$ and $\sf WPDL$ (Theorem \ref{theoEXPTIME}). \qedhere
\end{proof}

\begin{remark}
    We also obtain the exponential model property for both logics, but this is already known for $\sf CS4$~\cite{balbiani2024variants} and the proof techniques preserve infallibility, so they also apply to $\sf WS4$.
\end{remark}

\section{Concluding remarks}

We have introduced the constructive propositional dynamic logics of regular programs ${\sf CPDL}$ and its Wijesekera variant ${\sf WPDL}$, and established Gödel--Tarski translations into classical $\sf PDL$. As a consequence, both logics enjoy the exponential-size finite model property and have validity problems in {\sc ExpTime}.
By embedding the classical fragment ${\sf K}^*$ (with simply the programs $a$ and $a^*$) into the $[a], [a^*]$-fragment ${\sf CK}^*_\Box$, we establish matching lower bounds and show that ${\sf CK}^\Box$, ${\sf CPDL}_\Box$, ${\sf CPDL}$, and ${\sf WPDL}$ are all ${\sc ExpTime}$-complete.
In particular, this settles the conjecture of Afshari et al.~\cite{AfshariGLZ24} for ${\sf CK}^\Box$.

We have also translated $\sf CS4$ and $\sf WS4$ into our constructive variants of $\sf PDL$, thus also obtaining an {\sc ExpTime} upper bound, improving the previously known {\sc NExpTime} upper bound obtained via finite model arguments~\cite{balbiani2024variants}. 
The only lower bound we are currently aware of is {\sc PSpace}, inherited from the intuitionistic component~\cite{statman1979intuitionistic}.
Our conjecture is that these logics are in fact {\sc PSpace}-complete.

Finally, we remark that our logics are defined semantically and we have not
provided a deductive calculus. ${\sf CK}^*_\Box$ does enjoy sound and complete calculi~\cite{NishimuraConstructivePDL,celani,AfshariGLZ24}, and we also leave the development of Hilbert, Gentzen and/or cyclic
calculi for the rest of constructive variants of $\sf PDL$ as an open problem.

\end{document}